\documentclass[journal,twocolumn]{IEEEtran}

\usepackage{graphicx,epstopdf,color,amssymb,amsopn,amsthm,array,multirow,eqparbox,balance,hhline,enumerate,cite}

\usepackage{amsmath,mathtools}

\usepackage[caption=false,font=footnotesize]{subfig}
\usepackage{url}

\usepackage{algorithmic,amsthm}
\usepackage[multiple]{footmisc}

\usepackage[normalem]{ulem}
\usepackage{placeins}
\usepackage{booktabs}
\usepackage[ruled]{algorithm2e} 
\usepackage{threeparttable}

\newcommand\so{\bgroup\markoverwith{\textcolor{red}{\rule[0.5ex]{2pt}{0.4pt}}}\ULon}

\newtheorem{define}{Definition}
\newtheorem{proposition}{Proposition}
\newtheorem{remark}{Remark}
\newtheorem{theorem}{Theorem}
\newtheorem{lemma}{Lemma}

\def\X{\mathcal{X}}
\def\U{\mathcal{U}}
\def\D{\mathcal{D}}
\def\M{\mathcal{M}}
\def\Sm{\mathcal{S}}
\def\Om{\mathcal{O}}
\def\e{\varepsilon}
\def\Prob{\mathrm{Pr}}
\def\R{\mathbb{R}}

\def\Z{\mathbb{Z}}
\def\Lap{\mathrm{Lap}}
\def\Lm{\mathrm{Lap}_m}
\def\Ll{\mathcal{L}}
\def\Geo{\mathrm{Geom}}
\def\Gm{\mathrm{Geom}_m}

\def\ai{\alpha_{1i}}
\def\aii{\alpha_{2i}}
\def\Prd{\mathcal{P}}

 \def\L{\mathcal{L}} 
\newcommand{\ve}[1]{\boldsymbol{#1}} 
\newcolumntype{I}{!{\vrule width 1pt}}
\newlength\savedwidth
\newcommand{\whline}{\noalign{\global\savedwidth\arrayrulewidth
                            \global\arrayrulewidth 1.2pt}
                   \hline
                   \noalign{\global\arrayrulewidth\savedwidth}
                   }

\begin{document}
\IEEEoverridecommandlockouts
\title{More Flexible Differential Privacy: The Application of Piecewise Mixture Distributions in Query Release}
\author{David B. Smith, Kanchana Thilakarathna and Mohamed Ali Kaafar\vspace{-20pt}
\IEEEcompsocitemizethanks{\IEEEcompsocthanksitem David B. Smith, Kanchana Thilakarathna, Mohamed Ali Kaafar are with Data61-CSIRO, Eveleigh, NSW 2015, Australia. E-mail: \{firstname.lastname\}@data61.csiro.au\protect\\}}

\IEEEtitleabstractindextext{
\begin{abstract}

There is an increasing demand to make data ``open" to third parties, as data sharing has great benefits in data-driven decision making. However, with a wide variety of sensitive data collected, protecting privacy of individuals, communities and organizations, is an essential factor in making data ``open". The approaches currently adopted by industry in releasing private data are often ad hoc and prone to a number of attacks, including re-identification attacks, as they do not provide adequate privacy guarantees. While differential privacy has attracted significant interest from academia and industry by providing rigorous and reliable privacy guarantees, the reduced utility and inflexibility of current differentially private algorithms for data release is a barrier to their use in real-life. This paper aims to address these two challenges. First, we propose a novel mechanism to augment the conventional utility of differential privacy by fusing two Laplace or geometric distributions together. We derive closed form expressions for entropy, variance of added noise, and absolute expectation of noise for the proposed piecewise mixtures. Then the relevant distributions are utilised to theoretically prove the privacy and accuracy guarantees of the proposed mechanisms. Second, we show that our proposed mechanisms have greater flexibility, with three parameters to adjust, giving better utility in bounding noise, and mitigating larger inaccuracy, in comparison to typical one-parameter differentially private mechanisms. We then empirically evaluate the performance of piecewise mixture distributions with extensive simulations and with a real-world dataset for both linear count queries and histogram queries. The empirical results show an increase in all utility measures considered, while maintaining privacy, for the piecewise mixture mechanisms compared to standard Laplace or geometric mechanisms.
\end{abstract}

}


\maketitle
\IEEEdisplaynontitleabstractindextext
\IEEEpeerreviewmaketitle
\section{Introduction}
A number of \emph{open data platforms}\footnote{http://data.gov.au}\footnote{http://www.openiot.eu}  and Web API standards\footnote{http://www.opengeospatial.org/docs/is}\footnote{http://www.hypercat.io}  are enabling researchers and data analysts to provide revolutionary attractive services, such as travel updates, smart parking, and health monitoring, in a wide variety of domains. However, the sensitive nature of certain data categories, such as health care, is a barrier to evolving these data sharing platforms. Many government authorities and policy makers have been proactive in imposing rules and regulations to safeguard individual privacy and security in releasing data \cite{Weber2015618}. Privacy preserving online data release has then become more important than ever so as to enable the use of the data while not breaching individuals privacy.

There have been a number of mechanisms proposed to release private data that are primarily based on either data aggregation \cite{Acs2014}, data elimination \cite{el2009evaluating} or data anonymization \cite{bayardo2005data}. However, the majority of these mechanisms do not provide guaranteed privacy and they are often susceptible to re-identification attacks \cite{Zang2011, Xu2017}. The pervasive availability of external public data sources such as social networking data makes the problem of re-identification with linkage attacks even more acute \cite{Su2017}. There have been then many mechanisms based on adding noise to original query answers based on various bounded \cite{marley2011method}, and unbounded \cite{dp-book, li2010}, probabilistic distributions. Although bounded noise mechanisms provide some level of protection, privacy can not be guaranteed as the true value can be recovered with iterative queries. On the other hand, unbounded noise mechanisms such as Laplace or Gaussian noise can provide adequate privacy guarantees, but they often do not provide required accuracy to analysts \cite{dankar2013practicing}.

%
Differential privacy has gained much attention in the recent past as one of the unbounded noise mechanisms to provide private data release. Differential privacy gives provable guarantees for \emph{indistinguishability}, and hence privacy, of a particular participant's record, and sensitive information, in a database -- as there is no further privacy violation whether or not their information is in the database.
From the initial seminal work by Dwork et al. \cite{dwork2006}, adding noise from a zero-mean Laplace distribution to data queries to ensure robust privacy guarantees, via \emph{$\e$-differential privacy}, e.g., \cite{dp-book, li2010, hardt2012, fan2012, hay2010} has become common-place. The scale parameter of the Laplace distribution directly relates to the inverse of the privacy budget $\e$, where smaller $\e$ and greater scale parameter implies greater privacy, but less accuracy to the analyst, implying a direct trade-off between the two. Despite many realms of academic work to date and its promise in providing provable privacy guarantees, differential privacy has not yet been adopted by many in industry and government agencies primarily due to: i) Concerns over reduced utility of differential private query release \cite{dankar2013practicing}; and ii) Inherent to all the described mechanisms for differential privacy are a direct function of privacy budget $\e$, always with a fixed sensitivity, with little extra flexibility for the query-mechanism designer, in particular for linear queries.

In this paper, we propose a novel mechanism to draw noise for differential private query release based on piecewise mixture distributions, where two separate distributions with privacy budget parameter $\e$ and $r \e$, $r>0$ are fused together at a break-point $c_t$. The proposed solution comprises two Laplace distributions and two symmetric geometric (discrete Laplace) distributions respectively, for fractional query and integer query release. Our mechanisms provide increased utility benefits while keeping the same privacy guarantees provided by standard Laplace or geometric distributions. Moreover, the dataset curator is provided with three parameters, namely $\e$, $r \e$ and $c_t$, to control the tradeoffs of privacy loss, utility and accuracy as opposed to one parameter to manipulate in these typical differential privacy mechanisms -- while the desirable properties of the Laplace and geometric mechanisms are maintained. 
For brevity, we only investigate the benefits for integer counting and histogram queries, but the proposed mechanisms are obviously extensible to other types of queries such as linear fractional queries.

The paper makes the following contributions:
\begin{itemize}
\item First, we prove that any mechanism where privacy-seeking perturbations are kept within desired bounds such as drawing noise from truncated Laplace or geometric distributions, in many cases of linear queries do not preserve differential privacy.

\item Then, we propose two piecewise mixture distributions, as mixtures of Laplace and symmetric geometric mechanisms respectively, that provide guaranteed privacy with enhanced utility and more flexibility to the query designer.

\item We also derive closed-form expressions for the absolute-first and second moments of the proposed piecewise mixture mechanisms, as well as the entropy and present a new general privacy budget parameter $\zeta_{\e}$ relevant to piecewise mixture mechanisms, with privacy-preserving properties analogous to $\e$.

\item Finally, we evaluate the proposed mechanisms with extensive probabilistic simulations as well as with a real-world data set particularly suited to private linear querying. The results show that proposed piecewise mixture mechanisms provides better utility for the analysts compared to standard Laplace or geometric mechanisms.


\end{itemize}

The paper is organized as follows: Section II provides background on differential privacy and related work on providing increased utility. Then, in Section III, we propose piecewise mixture distribution mechanisms, defining them and providing their statistical properties, with derivations to analytically prove their privacy guarantees, as well as determining utility-privacy tradeoffs. We evaluate and compare their performance with traditional mechanisms by extensive analysis, simulations, and, importantly, with a real-world dataset in Section IV. Section V concludes the paper and provides some future directions.

%

\vspace{-14pt}

\section{Background on Differential Privacy}

Let $\U$ be the set of possible data points. A database of $n$ users contains the data points for each user and can be viewed as a vector in $\U^n$. Let $\D$ := $\U^n$  be the collection of databases with $n$ users. Two databases $\X^{(1)} \in \D$ and  $\X^{(2)} \in \D$  are called neighboring (denoted by $\X^{(1)} \sim \X^{(2)}$), if they differ by at most one coordinate.

\begin{define}
($\e$-differential privacy) A randomized mechanism (function)  $\M : \D \rightarrow \Om$ preserves
$\e$-differential privacy, if for any two neighboring databases $\X^{(1)} \sim \X^{(2)}$, and any possible set
of output $\Sm \subseteq \Om$, the following hold:
\begin{equation}
\Prob[\M(\X^{(1)})\in \Sm] \leq \exp(\e) \cdot \Prob[\M(\X^{(2)}) \in \Sm],
\end{equation}
where the randomness comes from the coin flips of $\M$\cite{dp-book}.
\end{define}

\begin{remark}
If $\Om$ is a countable set, then we can also have the inequality for each $x \in \Om$,
\begin{equation}
\Prob[\M(\X^{(1)})=x] \leq \exp(\e) \cdot \Prob[\M(\X^{(2)})=x],
\end{equation}
\end{remark}

\begin{define}
$\ell_1$ sensitivity. Let $f : \D \rightarrow \R^d$ be a deterministic function. The $\ell_1$-sensitivity of $f$, denoted by $\Delta f$, is
\begin{align}
\max&_{\X^{(1)}\sim \X^{(2)}} \|f(\X^{(1)}) - f(\X^{(2)})\|_1 =\nonumber\\
&\max_{\X^{(1)}\sim \X^{(2)}} \sum_{i=1}^d |f_i(\X^{(1)}) - f_i(\X^{(2)})|
\end{align}
\end{define}

The $\ell_1$, or global, sensitivity of a counting query and a histogram query $\Delta f=1$, since removing one user from $\X$ affects the outcome of the query by $1$ (in the case of histograms the cells or bins are disjoint).


It is widely known that adding noise from an appropriately scaled zero-mean Laplace distribution preserves differential privacy.
For discrete value querying, such as for linear counting queries and histogram queries, the discrete analog of the Laplace distribution, known as the symmetric geometric distribution \cite{inusah2006}, has been widely studied as also preserving $\e$-differential privacy, e.g., \cite{shi2011,barthe2013}. 

\subsection{Laplace mechanism}

The zero-mean Laplace distribution (a symmetric version of the exponential function) has the following probability density function (PDF):
\begin{equation}
\Lap(x|b)=\frac{1}{2b}\exp\left(\frac{-|x|}{b}\right).
\end{equation}
We denote a random variable drawn from a Laplace distribution with scale $b$ as $Y \sim \Lap(b)$. For a linear query, it has been widely demonstrated that adding noise from a zero-mean Laplace distribution with scale $b=\frac{\Delta f}{\e}$ satisfies $\e$-differential privacy.

The Laplace mechanism is
\begin{equation}
\M_L(\X)=f(\X)+(Y_1,\ldots,Y_k)
\end{equation}
where $Y_i$ are i.i.d. random variables drawn from $\Lap(\frac{\Delta f}{\e})$.

Hence, for any queries on neighboring databases $\X^{(1)}$, $\X^{(2)}$, then with privacy loss defined as
\begin{equation}
\L^\xi_{\M(X^{(1)})\|\M(X^{(2)})}=\ln \left(\frac{\Prob[\M(\X^{(2)})=\xi]}{\Prob[\M(\X^{(1)})=\xi]}\right)\label{ploss}
\end{equation}
then for this mechanism $\L^\xi \leq \ln\left( \exp\left(\frac{\e\|f(\X^{(1)})-f(\X^{(2)})\|_1}{\Delta f}\right)\right)$ so $\L^\xi \leq \e$, for any query output $\xi$, where $\ln(\cdot)$ denotes the natural logarithm.

\subsection{Geometric mechanism}
The Laplace mechanism adds real-number noise to any linear query, giving differential privacy, and for a general integer counting query a form of post-processing, to which $\e-$differential privacy is immune, the real-number noise can be rounded to the closest integer. It has been reported that drawing noise from a symmetric geometric distribution $\Geo\left(\exp\left(\frac{\e}{\Delta f}\right)\right)$ for linear query answering is the discrete (integer) analogue of the Laplace distribution, e.g., ~\cite{shi2011,barthe2013}, this reflects other literature which describes these as discrete Laplace distributions, e.g., \cite{inusah2006}. Thus, it could be expected that a random variable drawn from a continuous Laplace distribution could be transformed (e.g., rounded, floored) such that a symmetric geometrically distributed random variable results, but this is not the case. The symmetric geometric probability mass function at any integer $k$ is
\begin{equation}
\Geo(\alpha) = \left(\frac{\alpha-1}{\alpha+1}\right)\alpha^{-|k|}.
\end{equation}
Then if $\alpha=\exp\left(\frac{\e}{\Delta f}\right)$, $\e$-differential privacy is preserved.

\begin{remark}
The difference between a Laplace mechanism mapped to integers and the geometric mechanism (even though this is from the so-called discrete Laplace distribution) can be observed from the generation of random variables from their respective distributions.

Directly from the specification of the Laplace distribution, it is clear that $Y_{\Lap}$ can be generated from the difference of two i.i.d. exponentially distributed random variables, $\exp\{\cdot\}$ with parameter $\lambda = 1/b =\e / \Delta f$, $Y_{\Lap}\sim\exp_1\{\lambda\}-\exp_2\{\lambda\}$. Thus for integer count perturbation $\e$-differential privacy, the appropriate rounding gives $\lceil\exp_1\{\lambda\}-\exp_2\{\lambda\}\rfloor$ (similarly the floor $\lfloor\exp_1\{\lambda\}-\exp_2\{\lambda\}\rfloor$ and $\lceil\exp_1\{\lambda\}-\exp_2\{\lambda\}\rceil$ could equally be applied that would maintain differential privacy)

The geometric mechanism where $\alpha=\exp(\e/\Delta f)$ is also related to the exponential distribution. However, where $\lambda = \e / \Delta f$, then the geometric mechanism is specified by $Y_{\Geo}\sim\lfloor\exp_1\{\lambda\}\rfloor-\lfloor\exp_2\{\lambda\}\rfloor$ --- hence, it is immediately apparent that the rounded $\e$-private Laplace mechanism $\lfloor Y_{\Lap} \rceil$ is close to, but not the same as the $\e$-private geometric mechanism $Y_{Geom}$.
\end{remark}


\subsection{Mixture distribution mechanisms}
The application of distribution mechanisms and their prospective optimality led to the derivation of optimal mechanisms for a given privacy budget $\e$. These ``geometric staircase mechanisms'' for both differential privacy and approximate differential privacy are found in ~\cite{Optimal3} and ~\cite{Optimal1} respectively. There a staircase-shaped mechanism is proposed, as a mixture of uniform distributions, to optimize usefulness with respect to $\e$-differential privacy, as well as result for approximate differential privacy. Importantly \emph{for integer value queries, such as for histograms and counting queries, the staircase shaped mechanism reduces to the $\e$-differentially private symmetric geometric mechanism \cite{Optimal3}.}

\subsubsection{Truncated mechanisms}

It has recently been reported in \cite{Liu2016} that truncated Laplace mechanisms (as a special case of what are termed generalized Gaussian distributions) can preserve $\e$-differential privacy. But this is dependent on query release from two neighboring databases $\X^{(1)}$ and $\X^{(2)}$ sharing the same bounds. But in a more general (and practical) case, a truncated Laplace mechanism (or from any other suitable distribution) will \emph{not preserve differential privacy} or even approximate differential privacy. This is evident from the following example:

Consider queries from  $\X^{(1)}$ and $\X^{(2)}$ (and $\X^{(2)}$ has one more user) with $\ell_1$ sensitivity $\Delta f=1$, where we draw noise from a Laplace distribution with scale factor $b$, and we only perturb for query release by $\lceil Y \rfloor$, $Y \sim \Lap(b) $. Now the noise is only added if in some bound $|Y| \leq c$ for the two neighboring databases. Then consider the potential case $\lceil|Y^{(2)}|\rfloor=c$, and query output $\psi$ from $X^{(2)}$ with noise $c$ added, which has non-zero probability, then following from (\ref{ploss}) consider the privacy loss,
\begin{align}
\label{trunc_loss}
\Ll^\psi_{\M(X^{(2)})\|\M(X^{(1)})}=\ln \left(\frac{x_2[\textrm{where}\,x_2>0]}{x_1[\textrm{where}\,x_1=0]}\right) = \infty.
\end{align}
and the privacy loss $\L$ is unbounded, when $\M(X^{(1)}) \neq \M(X^{(2)})$, and there is no differential privacy.

For many cases of querying such as integer count querying, any truncated mechanism will have unbounded privacy loss, even though providing better utility to the standard unbounded Laplace mechanism. In the limit of privacy parameters in the piecewise mixture distributions we will introduce in the following sections, the piecewise mixture mechanism defaults to a standard Laplace, or geometric mechanism, either with less or more differential privacy, or a truncated Laplace or geometric mechanism.

\section{Differentially Private Mixture Mechanisms}

\subsection{Laplace Piecewise Mixture Mechanism}

We use the concept of the truncated Laplace mechanism to form what we call \emph{a Laplace piecewise mixture mechanism}, where there are effectively two scale parameters $b_1$ and $b_2$, where the scale parameter $b_2$ is applied around the origin  from $-c_t$ to $c_t$, and beyond a suitable point $c_t$ not far from the origin, the smaller scale parameter $b_1$ is applied -- effectively fusing two distributions. As we will show, the key aspect of the piecewise mixture mechanisms proposed is that the privacy loss, given any typical query, will vary between a lower privacy loss (with greater probability) and a higher privacy loss (with lower probability), and a greatly improved general privacy budget over standard mechanisms with similar accuracy performance for linear querying. Differential privacy is preserved and the dataset curator is provided with more \emph{flexibility} in mechanism design.

We set two privacy parameters $\e_r=r\e$ (beyond break-point $\pm c_t$) and $\e$ (within break-point $\pm c_t$) where $r > 0$ and typically $\e \leq 1$. A Laplace piecewise mixture distribution is denoted as $Y \sim \Lm\left(b_2=\frac{\Delta f}{\e},b_1=\frac{\Delta f}{r\e}\right)$, where set $\Delta f = 1$, with $\pm c_t$, at which we effectively fix and combine the two distributions $ Y \sim \Lap(1/\e)) $, $Y \sim \Lap(1/(r\e))$.

Note that as $r \rightarrow \infty$ then the relevant scale parameter $b_1 \rightarrow 0$ and the Laplace mixture mechanism tends toward a truncated Laplace mechanism. Note also that $r < \infty$ is required for the privacy loss $\L$ to remain bounded following from the $\L$ privacy loss description in the previous subsection.

\begin{define}
Hence the PDF for $Y \sim \Lm\left(\frac{1}{\e},\frac{1}{r\e}\right)$, $b_1=\frac{1}{r\e}$, $b_2=\frac{1}{\e}$,  can be formally given for all $x$ in $\R$ as
\begin{align}
\Lm(x | b_2, b_1) = \left\{\begin{array}{l l}\frac{a_1}{(2b_1)}\exp\left(\frac{-|x|}{b_1}\right),&|x|>c_t\\
\frac{a_2}{(2b_2)}\exp\left(\frac{-|x|}{b_2}\right),&|x|\leq c_t\end{array}\right. \label{PDFmix}
\end{align}
where $a_1$ and $a_2$ are constants related to the break-point $c_t$ and the scale parameters $b_1$ and $b_2$, designed to keep the cumulative distribution function (CDF) $\leq 1$, and for the PDF to be continuously defined,
\begin{align}
a_1 &\textstyle = \frac{p_1}{\left(p_1 \exp(-c_t/b_1) + p_2 (1- \exp(-c_t/b_2))\right)}\quad \textrm{and}\nonumber\\
a_2 &\textstyle = \frac{p_2}{\left(p_1 \exp(-c_t/b_1) + p_2 (1- \exp(-c_t/b_2))\right)},\quad \textrm{where}\\
p_1 &\textstyle = \frac{\exp(-c_t/b_2)}{b_2\left((1/(2b_1))\exp(-c_t/b_1)+(1/(2b_2))\exp(-c_t/b_2)\right)}\quad \textrm{and}\nonumber\\
p_2 &\textstyle = \frac{\exp(-c_t/b_1)}{b_1\left((1/(2b_1))\exp(-c_t/b_1)+(1/(2b_2))\exp(-c_t/b_2)\right)}.\nonumber
\end{align}
\end{define}
The CDF (which follows from the process to derive the standard zero-mean Laplace CDF, where we seek to keep the continuity of the CDF, specifically in the choice of $k_C$ below, such that there are no step changes) can then be expressed in closed form  as:
\begin{align}
{\Lm}_C(x | b_2, b_1)= \begin{cases}\frac{a_1}{2}\exp\left(\frac{x}{b_1}\right),&\scriptstyle x<-c_t \\
\frac{a_2}{2}\exp\left(\frac{x}{b_2}\right)+k_C,&\scriptstyle -c_t \leq x \leq 0\\
1-\left(\frac{a_2}{2}\exp\left(\frac{-x}{b_2}\right)+k_C\right),&\scriptstyle 0 < x \leq c_t \\
1-\frac{a_1}{2}\exp\left(\frac{-x}{b_1}\right),&\scriptstyle x>c_t\end{cases}\label{CDFmix}\\
\,\textrm{where}\, k_C=a_1/2\exp\left(-c_t/b_1\right)-a_2/2\exp\left(-c_t/b_2)\right).\nonumber
\end{align}
Key parameters for the Laplace piecewise mixture mechanisms, is (i). first absolute moment (expectation of noise); (ii). second moment (variance) and (iii). entropy. These can all be derived in closed form as follows:
(i). The expectation of the noise $E(|x|)$ can be derived as
\begin{align}
E_{\Lm}(|x|)&=\int_{-\infty}^{\infty} |x| \Lm(x) dx\nonumber\\
&= a_2\left(b_2-\exp\left(-c_t/b_2\right)(b_2+c_t)\right)\label{lapmnoise}\\
&+a_1\exp\left(-c_t/b_1\right)(b_1+c_t).\nonumber
\end{align}
(ii). The variance of the Laplace-piecewise distribution $\sigma^2$ can then be derived as
\begin{align}
\sigma^2_{\Lm}(x) &= \int_{-\infty}^{\infty} x^2 \Lm(x) dx\nonumber\\
&= 2a_2\left(b_2^2-\exp\left(-c_t/b_2\right)\left(b_2^2+b_2 c_t+c_t^2/2\right)\right)\label{varlapm}\\
&+2a_1\exp\left(-c_t/b_1\right)\left(b_1^2+b_1 c_t+c_t^2/2\right).\nonumber
\end{align}
(iii). The entropy $H_{\Lm}(x)$ can be derived as
\begin{align}
H&_{\Lm}(x) = -\int_{-\infty}^{\infty} \Lm(x)\ln(\Lm(x))dx\nonumber\\
&= \ln(2b_2/a_2)\left(1-a_1\exp\left(\frac{-c_t}{b_1}\right)\right)\nonumber\\ &\scriptstyle+\ln(2b_1/a_1)\left(a_1\exp\left(\frac{-c_t}{b_1}\right)\right)+\left(a_1/b_1\right)\exp\left(-\frac{c_t}{b_1}\right)(b_1+c_t)\label{lapment}\\
&-\left(a_2/b_2\right)\exp\left(-\frac{c_t}{b_2}\right)(b_2+c_t)+a_2.\nonumber
\end{align}
An example Laplace mixture PDF according to the definition of (\ref{PDFmix}), with two sets of scale parameters $\{b_2=10,b_1=1\}$ and $\{b_1=10,b_2=1\}$, and the two standard Laplace distributions for $b_{o,1}=1,\,b_{o,2}=10$ (i.e., $\e=1$ , $\e=1/10$) and $b_{o,3}=3.36$ (i.e., $\e=0.2975$) respectively are shown in  Fig. \ref{figLAPpdf}. Please note that the preferred implementation is for $b_2>b_1$, which is the $\Lm(b_2=10=1/\e,b_1=1/(r\e))$ dashed line in Fig. \ref{figLAPpdf}. For $\Lm(b_2=10,b_1=1)$ the effective fusing of distributions is apparent. Example Laplace mixture CDFs according to the expression in (\ref{CDFmix}) for the same mixture mechanisms and standard Laplace mechanisms are provided in Fig. \ref{figLAPcdf}. Here, for the preferred $\Lm(b_2=10=1/\e,b_1=1/(r\e))$, it is clear that the CDF rapidly tends to 0 and 1 beyond (below and above) the break-points of $-c_t$ and $c_t$ respectively.
\begin{figure}[tb]
\centering
\includegraphics[width=0.8\columnwidth]{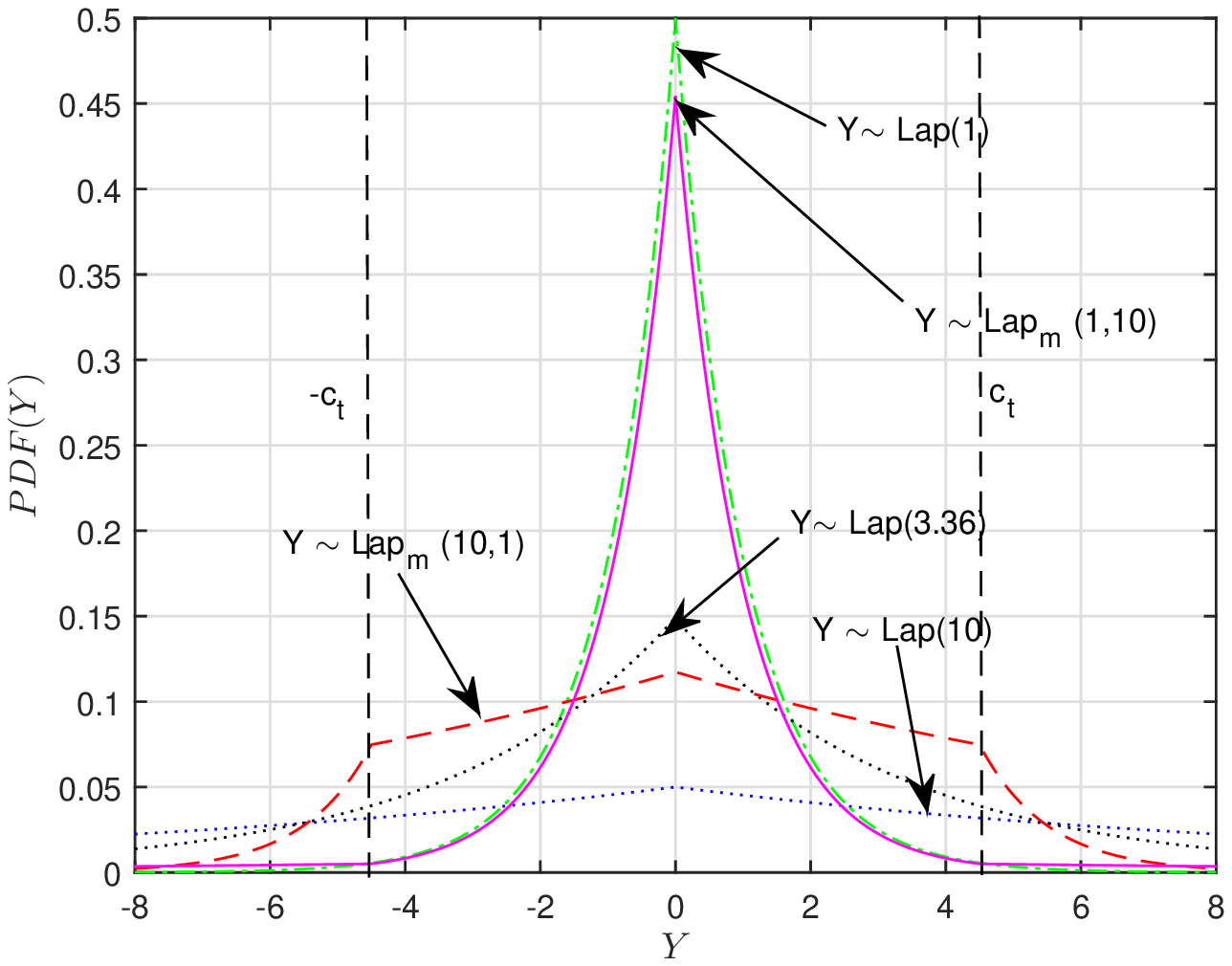}
\caption{Probability Density Function ($PDF(Y)$) for standard Laplace mechanisms with $\e=1/10$, $\e=1$, $\e=0.2975$ (respective to $b=1/\e=10$, $b=1$ and $b=3.36$)  and two Laplace mixture mechanisms $Y \sim \Lm(1/\e=10,1/(r\e)=1)$ and $Y \sim \Lm(1/\e=1,1/(r\e)=1/10)$ with breakpoint $c_t=4.5$ shown}
\label{figLAPpdf}
\end{figure}


\begin{figure}[tb]
\centering
\includegraphics[width=0.8\columnwidth]{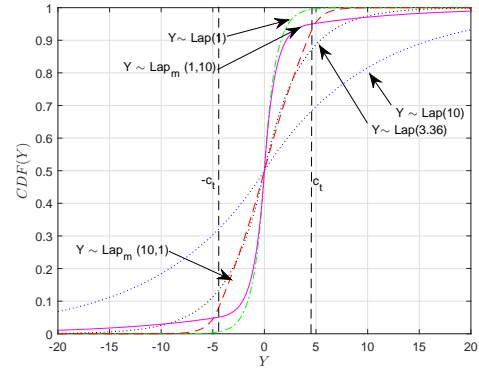}
\caption{Cumulative Distribution Function ($CDF(Y)$) for standard Laplace mechanisms with $\e=1/10$, $\e=1$, $\e=0.2975$, $b=1/\e$, and two Laplace mixture mechanisms $Y \sim \Lm(1/\e=10,1/(r\e)=1)$ and $Y \sim \Lm(1/\e=1,1/(r\e)=1/10)$ with breakpoint $c_t=4.5$ shown}
\label{figLAPcdf}
\end{figure}

Thus following from the closed form CDF we can simply generate $\Lm(b_2,b_1)$ distributed random variables $Y_{\Lm}$ according to Algorithm \ref{alg:mix}.

\begin{algorithm}[!t]
\small
\caption{Generating a Laplace-mixture random variable $Y \sim \Lm(b_2,b_1)$} \label{alg:mix}
\begin{algorithmic}[1]
\STATE Draw variable $r_u$ randomly from a uniform distribution in the range $[0,1]$, then
\IF{$r_u < (a_1/2)\exp(-c_t/b_1)$}
\STATE{$Y = b_1\ln(2r_u/a_1)$,}
\ELSIF{$r_u > 1-(a_1/2)\exp(-c_t/b_1)$}
\STATE{$Y = -b_1\ln(2(1-r_u)/a_1),$}
\ELSIF{$r_u \leq 1/2$,}
\STATE{$Y = b_2\ln\left(2/a_2(r_u-(a_1/2\exp(-c_t/b_1)-a_2/2\exp(-c_t/b_2)))\right)$}
\ELSE
\STATE{\footnotesize $Y = b_2\ln\left(2/a_2(1-r_u-(a_1/2\exp(-c_t/b_1)-a_2/2\exp(-c_t/b_2)))\right)$}
\ENDIF
\end{algorithmic}
\small
\end{algorithm}

For ease of nomenclature, in the remainder of the paper we refer to the Laplace mixture mechanism as $\Lm(\e,r\e)$ rather than $\Lm(1/\e,1/(r\e))$. Note that these nomenclatures are equivalent.

\subsubsection{Privacy Characteristics of Laplace Mixture Mechanism}
\begin{theorem}
The Laplace piecewise mixture mechanism is $\max\{\e,\e_r=r\e\}$ differentially private.
\end{theorem}
\begin{proof}
This can be immediately derived from the two pieces of the PDF with absolute value of noise $|Y|$ less-than-or-equal, or greater than the break-point $c_t$. Within $\pm 1$ of the break-point $c_t$ the privacy loss tends from $\e$ to $\e_r$ for neighboring databases $\X^{(1)}$ and $\X^{(2)}$. If $r>1$ then the piecewise mixture mechanism is $\e_r$ differentially private. Hence for true count $n_i$ from $f(\X^{(1)})$ or $f(\X^{(2)})$:

If $|Y^{(1)}|+n_i \leq n_i+c_t \wedge |Y^{(2)}|+n_i\leq n_i+c_t$
\begin{align}
{}^1\L_{\M(X^{(1)})\|\M(X^{(2)})}&=\ln\left( \exp\left(\frac{\e\|f(\X^{(1)})-f(\X^{(2)})\|_1}{\Delta f}\right)\right)\nonumber\\
&=\e,\nonumber
\end{align}
Else, if $|Y^{(1)}|+n_i > n_i+c_t \wedge |Y^{(2)}|+n_i> n_i+c_t$,
\begin{align}
{}^2\L_{\M(X^{(1)})\|\M(X^{(2)})}&=\ln\left( \exp\left(\frac{\e_r\|f(\X^{(1)})-f(\X^{(2)})\|_1}{\Delta f}\right)\right)\nonumber\\
&=\e_r,\nonumber
\end{align}
Else,
\begin{align}
{}^3\L_{\M(X^{(1)})\|\M(X^{(2)})}&=\ln\left( \exp\left(\frac{\e_\kappa\|f(\X^{(1)})-f(\X^{(2)})\|_1}{\Delta f}\right)\right)\nonumber\\
&=\e_\kappa.\nonumber
\end{align}
where $\min\{\e,\e_r\}<\e_\kappa<\max\{\e,\e_r\},$
\end{proof}

\begin{remark}
The privacy loss $\L(\M)={}^1\L$ from the CDF with probability $1-a_1\exp(-c_t/b_1)$ is bounded by $\e$ where typically $\e \ll \e_r$ and $ a_1\exp(-c_t/b_1) \ll 1$.
\end{remark}

Furthermore the Laplace mixture mechanism has the following property of accuracy: For query release an attribute that can take $k$ potential values, or alternatively considering $k$ i.i.d. random variables ${\Lm}_{,k}$ added to the true query data $f(\X)$,
\begin{theorem}
$\Lm(\e,\e_r)$ is $\ln\left(\frac{k a_1}{\delta}\right)\left(\frac{\Delta f}{\e_r}\right)$ useful when $|{\Lm}(\e,\e_r)| > c_t$.
\end{theorem}
Approximately $r/\ln(a_1)$ more useful, i.e., more accurate, than a differentially private Laplace mechanism with privacy budget $\e$, where $\delta$ is some small number close to zero.

\begin{proof}
\begin{align}
&\Prob\left[\|f(\X)-\M_L(\X)\|_{\infty} \geq \ln\left(\frac{k a_1}{\delta}\right)\left(\frac{\Delta f}{\e_r}\right)\right]\nonumber\\
&=\Prob\left[\max_{i \in [k]} |Y_i| \geq \ln\left(\frac{k a_1}{\delta}\right)\left(\frac{\Delta f}{\e_r}\right)\right]\\
&\leq k \Prob\left[ |Y_i| \geq \ln\left(\frac{k a_1}{\delta}\right)\left(\frac{\Delta f}{\e_r}\right)\right]\nonumber\\
&=k\left(\frac{\delta}{k a_1}\right)\exp\left(\ln(a_1)\right),\nonumber\\
&=\delta.\nonumber
\end{align}
\end{proof}

\subsubsection{Accuracy/privacy tradeoff by cost formulation}

One measure follows from ~\cite{Optimal3,ghosh2012,brenner2010,gupte2010}, where the combined utility with respect to particular accuracy-losses can be combined in the following metric for expectation of cost of $x$
\begin{equation}
\mathrm{Cost_{per.\,budget}}=\chi=\int_{-\infty}^\infty\Ll(x)\Prd(x) dx.
\end{equation}
One better quantification of accuracy-loss $\Ll(x)$, is the absolute value of the noise $x$ added hence $\Ll(x)=|x|$, which gives the expectation of the noise amplitude. Then if the probability distribution at $x$, $\Prd(x)$ is specified by the zero-mean Laplace distribution, as per the Laplace mechanism, this integral simply equals to $\frac{\Delta f}{\e} = 1/\e$ for sensitivity one queries. The expression for $\chi$ where  $\Ll(x)=|x|$ for the piecewise mixture Laplace mechanism with respect to $r$,$\e$ and the breakpoint at $c_t$ is given in (\ref{lapmnoise}).

Another quantification of accuracy-loss is the variance of the noise where $\Ll(x)=x^2$, which for the zero-mean Laplace distribution has a value of $2b^2 = 2/\e^2$ for sensitivity one queries. The expression for $\chi$ where $\Ll(x)=x^2$ for the piecewise mixture Laplace mechanism is given in (\ref{varlapm}).


\vspace{-5pt}

\subsection{Geometric Piecewise Mixture Mechanism}

We now provide the piecewise mixture mechanism formed from fusing two discrete Laplace distributions of different scale parameters $\alpha_1=\exp(r\e)$, $\alpha_2=\exp(\e)$ around a break-point $c_t$, in a similar manner as applied to continuous Laplace distributions.

\begin{define}
For $Y \sim \Gm\left(\alpha_2,\alpha_1\right)$, $\alpha_1=\exp(r\e)$, $\alpha_2=\exp(\e)$, where we set $\Delta f = 1$, the probability mass function (PMF) of the piecewise mixture geometric mechanism can be formally given for all $x$ in $\Z$, $c_t$ in $\Z^{+}$ as
\begin{align}
\Gm(x | \alpha_1, \alpha_2) = \left\{\begin{array}{c l}a_{1,g}\left(\frac{\alpha_1-1}{\alpha_1+1}\right)\alpha_1^{-|x|},&|x|>c_t\\
a_{2,g}\left(\frac{\alpha_2-1}{\alpha_2+1}\right)\alpha_2^{-|x|},&|x|\leq c_t,\end{array}\right.\label{PDFgmix}
\end{align}
where
\begin{align}
 a_{1,g} &\textstyle= \frac{g_1}{g_1 \alpha_1^{-c_t} + g_2 (1 - \alpha_2^{-c_t})}\nonumber\\
a_{2,g} &\textstyle = \frac{g_2}{g_1 \alpha_1^{-c_t} + g_2 (1 - \alpha_2^{-c_t})}\label{DFgmix}\\
 g_1 &\textstyle=2\frac{\Geo(\alpha_2,x=-c_t)}{\Geo(\alpha_2,x=-c_t)+\Geo(\alpha_1,x=-c_t)}\nonumber\\
g_2 &\textstyle= 2\frac{\Geo(\alpha_1,x=-c_t)}{\Geo(\alpha_2,x=-c_t)+\Geo(\alpha_1,x=-c_t)}.\nonumber
\end{align}
\end{define}

An example geometric mixture PMF according to the definition of (\ref{PDFgmix}), with two sets of parameters $\{\alpha_2,\alpha_1\}=\{\exp(1/10),\exp(1)\}$ and $\{\exp(1),\exp(1/10)\}$, and the two standard geometric distributions for $\alpha=\exp(1),\,\alpha=\exp(1/10)$ (i.e., $\e=1$ , $\e=1/10$) and $\alpha=\exp(0.2894)$ respectively are shown in  Fig. \ref{figGEOpdf}. Please note that the preferred implementation is for $\alpha_2 < \alpha_1$, which is for $\Gm(\exp(1/10),\exp(1))$ in Fig. \ref{figGEOpdf}. For $\Gm(\exp(1/10),\exp(1))$ the effective fusing of distributions is apparent.
\begin{figure}[tb]
\centering
\includegraphics[width=0.8\columnwidth]{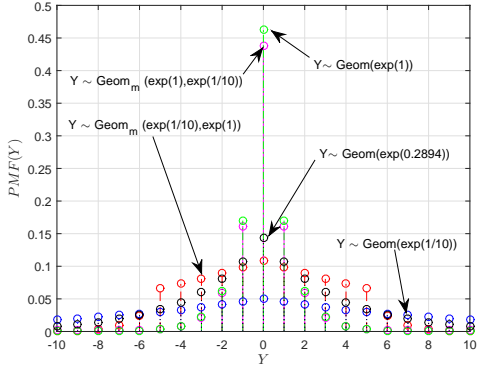}
\caption{Probability Mass Function ($PMF(Y)$) for standard geometric mechanisms with $\e=1/10,\e=1$ and $\e=0.2894$, and geometric mixture mechanisms $Y \sim \Gm(\exp(1/10),\exp(1))$, $Y \sim \Gm(\exp(1),\exp(1/10))$, where $r=10,\,1/10$ respectively, with breakpoint $c_t=5$ shown}
\label{figGEOpdf}
\end{figure}

Then for all $x$ in $\R$, the CDF of the piecewise geometric mixture (derived in a similar manner to the CDF of the piecewise Laplace mixture and the choice of $k_C$ below) is
\begin{align}
{\Gm}_C&(x | \alpha_1, \alpha_2)= \left\{\begin{array}{ll} a_{1,g} \frac{\alpha_1^{-\lceil x\rceil}}{\alpha_1+1},&\scriptstyle x <-c_t\\
a_{2,g} \frac{\alpha_2^{-\lceil x\rceil}}{\alpha_2+1}+k_C,&\scriptstyle -c_t\leq x <0 \\
1-a_{2,g} \frac{\alpha_2^{-\lceil x+1\rceil}}{\alpha_2+1}+k_C,&\scriptstyle 0\leq x \leq c_t \\
1-a_{1,g} \frac{\alpha_1^{-\lceil x+1\rceil|}}{\alpha_1+1},&\scriptstyle x > c_t\end{array}\right.\label{PCDFgmix}\\
&\textrm{where}\, k_C=a_1/2\alpha_1^{-c_t}-a_2/2\alpha_2^{-c_t}\nonumber
\end{align}

(i). The expectation of noise $E(|x|)=\chi$, which we are seeking to minimize with respect to privacy budget, for the geometric piecewise mixture is
\begin{align}
\chi_{\Gm}&=2\sum_{x=1}^{\infty}x\Gm(x)\nonumber\\
&=2a_{1,g}\frac{(c_t\alpha_1-c_t-1)\alpha_1^{c_t+1}}{\alpha_1^2-1}+\nonumber\\
&2a_{2,g}\frac{\left((c_t\alpha_2-c_t-1)\alpha_2^{c_t+1}+\alpha_2\right)}{\alpha_2^2-1}\label{Gmixnoise}
\end{align}

(ii). The variance of the geometric mixture distribution, $\sigma^2=\chi=E(x^2)$, which we are seeking to minimize, is
\begin{align}
\sigma&^2_{\Gm}=2\sum_{x=1}^{\infty}x^2\Gm(x)\nonumber\\
&=2a_{1,g}\frac{\alpha_1^{c_t+1}\left(c_t^2\alpha_1^2-(2c_t^2+2c_t-1)\alpha_1+(c_t+1)^2\right)}{(\alpha_1-1)^2}\label{Gmixvar}\\
&\textstyle+2a_{2,g}\frac{\alpha_2(\alpha_2+1)-\alpha_2^{c_t+1}\left(c_t^2\alpha_2^2-(2c_t^2+2c_t-1)\alpha_2+(c_t+1)^2\right)}{(\alpha_2-1)^2(\alpha_2+1)}.\nonumber
\end{align}

(iii). The entropy of the geometric mixture distribution, $H_{\Gm}(x)$ can be found similarly as
\begin{align}
H&_{\Gm}=-\sum_{x=-\infty}^{\infty}{\Gm(x)\ln(\Gm(x))}\nonumber\\
&\textstyle=-a_{1,g}\ln{a_{1,g}}\alpha_1^{-\lceil c_t\rceil}\ln\left(\frac{\alpha_1-1}{\alpha_1+1}\right)\nonumber\\
&\textstyle+2a_{1,g}\frac{(c_t\alpha_1-c_t-1)\alpha_1^{c_t+1}}{\alpha_1^2-1}\ln(\alpha_1)\nonumber\\
&\textstyle+2a_{2,g}\ln{\alpha_2}\nonumber\\&\textstyle\times{\frac{\alpha_2(\alpha_2+1)-\alpha_2^{c_t+1}\left(c_t^2\alpha_2^2-(2c_t^2+2c_t-1)\alpha_2+(c_t+1)^2\right)}{(\alpha_2-1)^2(\alpha_2+1)}}\label{Gmixent}\\
&\textstyle-a_{2,g}\ln{a_{2,g}} \left(1-\alpha_1^{-\lceil c_t\rceil}\right)\ln\left(\frac{\alpha_2-1}{\alpha_2+1}\right).\nonumber
\end{align}


Thus following from the closed form CDF we can simply generate $\Gm(\alpha_2,\alpha_1)$ distributed random variables $Y_{\Gm}$ according to Algorithm \ref{alg:gmix} below.

\begin{algorithm}[!b]
\small
\caption{Generating a geometric-mixture random variable $Y \sim \Gm(\alpha_2,\alpha_1)$} \label{alg:gmix}
\begin{algorithmic}[1]
\STATE $\ai=\alpha_1^{-1},\,\aii=\alpha_2^{-1}$.
\STATE Draw variable $r_u$ randomly from a uniform distribution in the range $[0,1]$, then
\IF{$r_u < a_{1,g}\frac{\ai^{c_t}}{1+\ai}$}
\STATE{$Y = \lceil \ln((1+\ai)r_u/a_{1,g})/(r\e) \rceil$,}
\ELSIF{$r_u > 1-a_{1,g} \frac{\ai^{c_t+1}}{1+\ai}$}
\STATE{$Y = \lceil -\ln((1-r_u)(1+\ai)/a_{1,g})/(r\e)-1\rceil$,}
\ELSIF{$r_u \leq a_{2,g}\frac{1}{1+\aii}+k_c$,}
\STATE{$Y = \lceil \ln((1+\aii)(r_u-k_c)/a_{2,g})/\e \rceil$,}
\ELSE
\STATE{$Y = \lceil -\ln((1-r_u-k_c)(1+\aii)/a_{2,g})/\e-1\rceil$,}
\ENDIF
\end{algorithmic}
\small
\end{algorithm}

For ease of nomenclature, in the remainder of the paper we refer to the geometric mixture mechanism as $\Gm(\e,r\e)$ rather than $\Gm(\exp(\e),\exp(r\e))$, as
we have referred to it in this section. Please note that these nomenclatures are equivalent.

\subsubsection{Privacy Characteristics of Geometric Mixture Mechanism}


\begin{theorem}
The geometric piecewise mixture mechanism is $\max\{\e,\e_r=r\e\}$ differentially private.
\end{theorem}
\begin{proof}
Below the break-point $c_t$ of this mechanism the privacy loss is $\e$, and above $c_t$ the privacy loss is $\e_r$, due to the mechanism generating noise from space of integers, the loss is either $\e$ or $\e_r$.
\end{proof}

\begin{remark}
From the CDF, with probability $1-a_1\alpha_1^{(-c_t)}$, the privacy loss $\L(\M)$ is bounded by $\e$ where typically $\e \ll \e_r$ in the preferred implementation of the mechanism, for $|x|\leq c_t$.
\end{remark}

Furthermore the geometric mixture mechanism has the following property of accuracy: for query release an attribute that can take $k$ potential values, or alternatively considering $k$ i.i.d. random variables ${\Gm}_{,k}$ added to the true query data $f(\X)$,
\begin{theorem}
$\Gm(\e,\e_r)$ is ${\Gm}_{,u} = \ln\left(\frac{k a_1}{\delta}\right)\left(\frac{\Delta f}{\e_r}\right)$ useful, where $
\ln\left(\frac{a_1 k}{\delta}\right)$ is a positive integer (in $\Z^{+}$) multiple of $\e_r$ when $|{\Gm}(\e,\e_r)| \geq c_t$
\end{theorem}
Approximately $r/\ln(a_1)$ more useful, i.e., more accurate, than a differentially private geometric mechanism with privacy budget $\e$.
\vspace{-15pt}
\begin{proof}
\begin{align}
&\Prob\left[|Y_i| \geq p/\e_r\times\e_r\right] = \frac{1}{\exp(p)}\,\textrm{where}\,p\in \Z^{+}\nonumber\\
&\Prob\left[\|f(\X)-\M_L(\X)\|_{\infty} \geq \ln\left(\frac{k a_1}{\delta}\right)\left(\frac{\Delta f}{\e_r}\right)\right]\nonumber\\
&=\Prob\left[\max_{i \in [k]} |Y_i| \geq \ln\left(\frac{k a_1}{\delta}\right)\left(\frac{\Delta f}{\e_r}\right)\right]\\
&\leq k \Prob\left[ |Y_i| \geq \ln\left(\frac{k a_1}{\delta}\right)\left(\frac{\Delta f}{\e_r}\right)\right]\nonumber\\
&=k\left(\frac{\delta}{k a_1}\right)\exp\left(\ln(a_1)\right),\nonumber\\
&=\delta\,\textrm{where}\,\delta\in\frac{k a_1}{\exp(\e_r p_d)},\, s.t. \{\exp(\e_r p_d) > k a_1\} \cup \{p_d  \in \Z^{+}\}\nonumber
\end{align}
\end{proof}
\vspace{-15pt}

To account for the variation of privacy parameters $\e$ and $r\e$ around the break-point $c_t$, for comparisons between standard and piecewise mixture mechanisms, and to evaluate the general accuracy vs. privacy tradeoffs, we next introduce the concept of a \emph{general privacy budget}, which has a very natural definition, as well as being useful for calculating the real privacy for rounded mechanisms, such as rounding the Laplace mechanism.
\vspace{-14pt}
\subsection{General Privacy Budget}

\begin{define}
Here we define a general privacy budget $\zeta_{\e}$, for neighboring databases $X^{(1)}$,$X^{(2)}$ differing by one coordinate, and where $\M(X^{(2)})\neq M(X^{(1)})$, then
\begin{align}
\zeta_{\e} &= \ln\left\{\sum_{\forall \xi} \exp\left(|\L^\xi_{\M(X^{(2)})\|\M(X^{(1)})}|\right)\Prob[\M(\X^{(1)})=\xi]\right\}\label{zetagen}\\
&\approx \ln\left\{\sum_{\forall \xi} \exp\left(|\L^\xi_{\M(X^{(2)})\|\M(X^{(1)})|}|\right)\Prob[\M(\X^{(2)})=\xi]\right\}\nonumber
\end{align}
Where $|\cdot|$ is absolute value. This clearly equals $\e$ for any geometric or Laplace $\e$-differentially private mechanism, as $|\L^\xi_{\M(X^{(2)})\|\M(X^{(1)})}|=\e$ $\forall \xi \in \Z$.
\end{define}
For queries from the piecewise geometric mechanism, with $\ell_1$-sensitivity $\Delta f = 1$, this is
\begin{align}
\exp\{\zeta_\e\}&=\exp(\e)\Prob[Y_{\Gm}\leq c_t] + \exp(r\e)\Prob[Y_{\Gm} > c_t]\nonumber\\
\zeta_{\e} & = \ln\left(\exp(\e)(1-a_{1,g} \exp(-r\e c_t))\right.\nonumber\\ &\left. +\exp(r\e) a_{1,g} \exp(-r\e c_t)\right)\label{zetageom}\\
&= \ln\left( a_{1,g} (\exp(r\e) - \exp(\e)) + \exp(r\e c_t + \e)\right) - r\e c_t \nonumber\\
&=  \e-r\e c_t + \ln\left( a_{1,g} (\exp((r-1)\e)-1)+\exp(r\e c_t)\right).\nonumber
\end{align}

For standard Laplace mechanism, with $\Delta f = 1$, where noise is rounded to the nearest integer for, e.g., integer count querying and histogram querying, we find $\zeta_{\e}$ as
\begin{align}
\exp&\{\zeta_\e\}=\frac{\Prob[|Y_{\Lap}|< 0.5)]^2}{\Prob[-0.5 \geq Y_{\Lap}< -1.5]}+\Prob[|Y_{\Lap}|< 0.5])\nonumber\\
&+\exp(\e)\left(\Prob[Y_{\Lap}\leq -1.5]+\Prob[Y_{\Lap}\geq 0.5]\right)\nonumber\\
\zeta&_\e=\ln\left\{\frac{(1-\exp(-0.5\e))^2}{0.5\exp(-0.5\e)-0.5\exp(-1.5\e)}\right.\nonumber\\ &\left.+\left(1-\exp(-0.5\e)\right)\right.\label{zetal}\\
&+\left.\exp(\e)(0.5\exp(-1.5\e)+0.5\exp(0.5\e))\right\}.\nonumber
\end{align}

For the Laplace piecewise mixture mechanism, for counting queries with sensitivity $\Delta f = 1$, we find that
\begin{align}
\exp&\{\zeta_\e\}=\frac{\Prob[|Y_{\Lm}|< 0.5)]^2}{\Prob[-0.5 \geq Y_{\Lm}< -1.5]}+\Prob[|Y_{\Lm}|< 0.5])\nonumber\\
&+\exp(\e)\left(\Prob[-c_t\leq Y_{\Lm}\leq -1.5]+\Prob[c_t \geq Y_{\Lm}\geq 0.5]\right)\nonumber\\&+\exp(r\e)\Prob[|Y_{\Lm}|>c_t]\nonumber\\
&\zeta_\e=\ln\left\{\frac{(1-a_2\exp(-0.5\e)-2k_C)^2}{0.5a_2\exp(-0.5\e)-0.5a_2\exp(-1.5\e)}\right.\nonumber\\&\left.+\left(1-a_2\exp(-0.5\e)-2k_C\right)\right.\label{zetalp}\\
&\left.+\exp(\e)a_2\left(0.5\exp(-1.5\e)+0.5\exp(0.5\e)-\exp(-c_t\e)\right)\right.\nonumber\\&\left.+a_1\exp(-r\e (c_t-1))\right\}.\nonumber
\end{align}


For a true count $f(\X)$, we assume that $f(\X)\gg 1$ in (\ref{zetageom}) and (\ref{zetalp}) (When $f(\X)$ is close to zero (\ref{zetageom}) and (\ref{zetalp}) are approximations). $\zeta_{\e}$ is bounded between $\e$ and $r\e$, and from (\ref{zetageom}) and (\ref{zetalp}) if $r \gg 1$ and $c_t > 2$ then the general privacy budget $\zeta_{\e}$ is greater than $\e$, but closer to $\e$ with less privacy loss. This is the preferred implementation of the piecewise mixture mechanisms. \footnote{If, alternatively, $r \ll 1$, then $\zeta_{\e}$ is less than $\e$ but closer to $\e$ than $r\e$ with a greater privacy loss.}

\begin{proposition}
The general privacy budget applies under composition, if $k$ piecewise mixture mechanisms have a general privacy budget of $\zeta_{\e,i}$ then their combination has a combined general privacy budget of $\sum_{i=1}^k \zeta_{\e,i}$. This even applies when there are different break-point bounds, $c_{t,i}$ for each $i$, as well as when there are separate $\e_i$ and/or $r_i \e_i$.
\end{proposition}

\begin{proof}
As for standard $\e-$differential privacy, for the combination of $k$ piecewise mixture mechanism, the overall general privacy budget is $\ln(\prod_{i=1}^k \exp(\zeta_{\e,i})) = \sum_{i=1}^k \zeta_{\e,i}$ for $\{\M_1(x),\M_2(x),\ldots,\M_k(x)\}$. This reduces to $\sum_{i=1}^k \e_i$ for standard mechanisms.
\end{proof}

\begin{remark}
\emph{Iterative online querying:} The above proposition for general privacy budget composition indicates the value of the piecewise mixture mechanism for iterative querying solutions, such as for large-scale querying using mechanisms such as online multiplicative weights \cite{hardt2010,hardt2012}, where a queried dataset is continuously updated. For the same general level of privacy, the rate at which the dataset is updated, such as for online multiplicative weights can be reduced according to a particular $\e-$related privacy threshold, as the ``do-nothing" case becomes more frequent, due to the perturbed-noisy query answer occurring more often within a given test-threshold for updates. Or, equivalently, the threshold can be tightened with greater general level of privacy, for the same number of updates as standard online multiplicative weights (or even offline algorithms such as ``dual query" \cite{gaboardi2014}).
\end{remark}

Furthermore with respect to general privacy-budget:
\begin{lemma}
If any mechanism provides $\e$-differential privacy, then $\zeta_{\e}$ exists and is bounded by $\e$. And, conversely, if $\zeta_{\e}$ exists and is bounded, then the corresponding mechanism will provide $\max\left(|\L^\xi_{\M(X^{(2)})\|\M(X^{(1)})}|\right)=\e$-differential privacy.
\end{lemma}

\begin{proof}
This key lemma follows directly from the definition of general privacy budget in (\ref{zetagen}) by either operations on the left-hand side, or right-hand side, of those equations. If $\M(\cdot)$ is $\e$-differentially private then privacy loss, $\L^\xi_{\M(X^{(2)})\|\M(X^{(1)})} \leq\e; \forall \xi$, and since $\Prob[\M(\X^{(1)})=\xi]$ and $\Prob[\M(\X^{(2)})=\xi]$ are only defined in region $[0,1]$, then $\zeta_{\e}$ must exist and be bounded by $\e$. As for the converse case, if it is known that the general privacy budget exists and is bounded, then it must be that that where $\e=\max\left(|\L^\xi_{\M(X^{(2)})\|\M(X^{(1)})}|\right) \; \forall \xi$, the relevant mechanism $\M(\cdot)$ is $\e$-differentially private.
\end{proof}

\section{Performance Evaluation}

\subsection{Analytical Utility and Privacy Evaluation}

Here we provide values of metrics for general privacy budget $\zeta_\e$ according to (\ref{zetageom}), (\ref{zetalp}) for geometric and Laplace-mixture piecewise mechanisms, given the three parameters of break-point $c_t$, $\e$ parameter value within break-point and $r\e$ parameter value above breakpoint. We also provide equivalent privacy budget $\e$ for the Laplace mechanism such that $\zeta_\e$ for the rounded Laplace mechanism according to (\ref{zetal}) is equal to that of (\ref{zetalp}). The equivalent privacy budget $\e$ for the standard geometric mechanism is simply equal to (\ref{zetageom}). We use metrics for accuracy-loss $\chi = E(|x|), \chi = \sigma^2(x)$ and entropy $H(x)$ for Laplace as well as for standard geometric mechanisms, and Laplace mixture mechanism according to  (\ref{lapmnoise}),(\ref{varlapm}) and (\ref{lapment}) respectively; and for geometric mixture mechanism according to (\ref{Gmixnoise}), (\ref{Gmixvar}) and (\ref{Gmixent}). This is summarized over a large range of these relevant parameters, with $\Delta f=1$ in Table I in the Appendix.

From Table I, in the Appendix, it is clear that significant benefits in terms of reduced loss for the mixture mechanisms, with less expected noise and less variance, as well as less entropy, are achieved over a wide range of $\e$ and $r\e$, and across the range $4\leq c_t\leq 7$ considered, with benefits from 0.1 to 0.5 for $\e$ (across the range of $\e$ investigated), and for $r\e<2$ when $r > 1$. For instance, for $c_t=5,\e=0.2,r\e=1$ the general privacy budgets $\zeta_\e$  are approximately equivalent, 0.3 for all mechanisms, and for mixtures compared with standard mechanisms: the expectation of noise $E(|x|)$ is approximately less by 0.5, the variance is a factor of 2 smaller, and the entropy is reduced by around 10\%. In Table I, in the Appendix, it is noteworthy that the greatest relative improvements for the mixture mechanisms are in terms of variance $\chi = \sigma^2(x)$ (in many cases less than half the variance of the standard mechanisms), but the improvements in reduced expected noise $E(|x|)$, and lower entropy $H(x)$, are also significant.

We also plot the two metrics of loss, as well as entropy for geometric mixture $\Gm(\e,r\e)$ and a breakpoint $c_t=5$, with a range of two mixture parameters $r\e$ ($r \geq 1$), $\e$ and provide the difference from the standard geometric mechanism with $\e_{\Geo}=\zeta_\e\{\Gm\}$. Positive values for this difference indicate superior performance of the geometric mixture mechanism. The first is expectation of noise $\chi=E(|x|)$, according to (\ref{Gmixnoise}), and we plot this as well as the difference in expectation of noise  $\chi_d=E(|x|)_{\Geo}-E(|x|)_{\Gm}$ in Fig. \ref{fig:enoise1}. We also present the log of variance $\log\chi=\log\sigma^2(x)$, $\sigma^2(x)$ from (\ref{Gmixvar}) and plot this in Fig. \ref{fig:evar1}, along with the difference in log of variance from the standard geometric mechanism with $l\chi_d=\log\sigma^2(x)_{\Geo}-\log\sigma^2(x)_{\Gm}$. Finally we present the entropy $H(x)$ and difference in entropy from the standard geometric mechanism $H_d(x)=H(x)_{\Geo}-H(x)_{\Gm}$ respectively in Fig. \ref{fig:eent1}. In Fig. \ref{fig:enoise1}, Fig. \ref{fig:evar1} and Fig. \ref{fig:eent1} the area of benefit in terms of each respective metric is provided where, as $\e$ decreases, even with increasing $r\e$ the geometric mixture outperforms the standard geometric mechanism. These figures also show benefit for $\e$ up to 1, with the maximum beneficial values of $r\e$, $r>1$, decreasing as $\e$ increases.

\begin{figure}[t!]
\centering\vspace{-4mm}
\includegraphics[width=0.92\columnwidth]{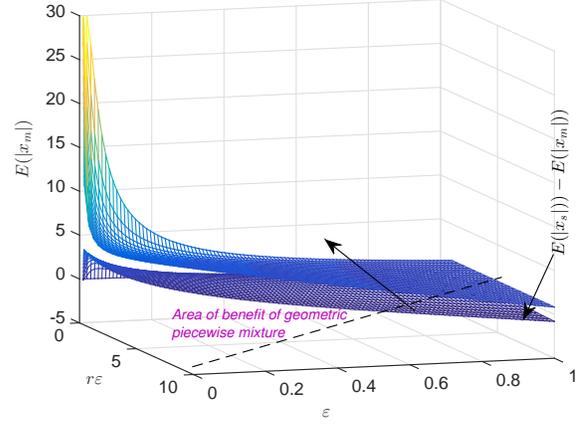}
\caption{Expectation of noise for geometric mixture, $E(|x|)$ versus two mixture parameters $r\e$,$\e$ and difference from standard geometric mechanism with $\e=\zeta_\e\{\Gm\}$, $c_t=5$}
\label{fig:enoise1}
\end{figure}

\begin{figure}[t!]
\centering
\includegraphics[width=0.92\columnwidth]{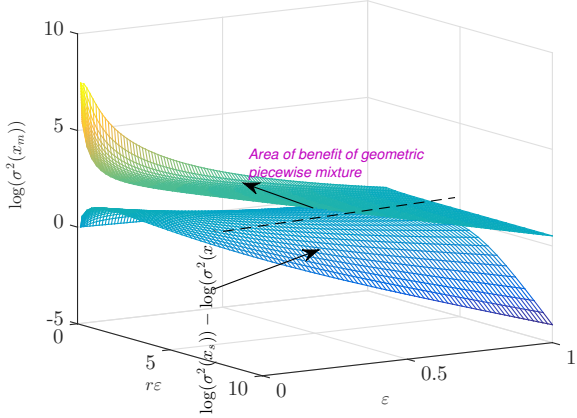}\vspace{-4mm}
\caption{Log of Variance of noise, $\log\sigma^2x$ for geometric mixture versus two mixture parameters $r\e$,$\e$ and difference of log of variance from standard geometric mechanism with $\e=\zeta_\e\{\Gm\}$,$c_t=5$}\vspace{-1mm}
\label{fig:evar1}
\end{figure}

\begin{figure}[t!]
\centering\vspace{-4mm}
\includegraphics[width=0.92\columnwidth]{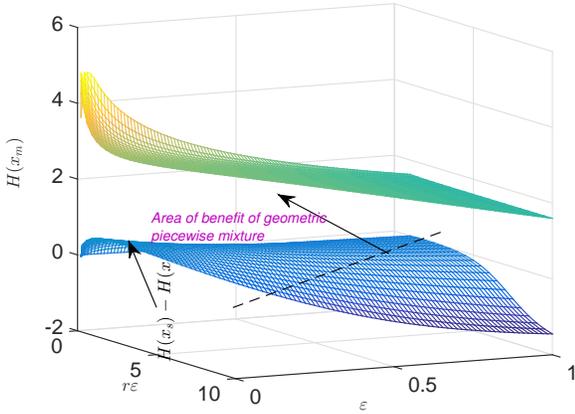}\vspace{-4mm}
\caption{Entropy of noise for geometric mixture, $H(x)_{\Gm}$ versus two mixture parameters $r\e$,$\e$ and difference of entropy from standard geometric mechanism with $\e=\zeta_{\e}\{\Gm\}$, $c_t=5$}\vspace{-4mm}
\label{fig:eent1}
\end{figure}

\vspace{-3pt}
\subsection{Utility Evaluation by Simulation}
\label{sec:metrics}
Here we test the performance of the proposed piecewise mixture mechanisms and compare with standard Laplace and geometric mechanisms by simulation.
\subsubsection{Utility Metrics for Simulation}
Having already investigated general privacy budgets, and compared with typical privacy budgets, we now seek to provide further suitable measures of utility. The first metric is the \emph{empirical CDF} of the error, $Y_i$ for geometric and Laplace mixtures and standard geometric, which is simply  (using $\leq$ because we are using integer counts) $\Prob(|n_i-y_i| \leq y_t)$, where $y_t \in \{\Z^{+},0\}$

The second metric is mean relative error, when above break-point $c_t$, as a weighted expectation of the added noise being greater than $c_t$ relative to the true count $n_i$ in the dataset
\begin{equation}
\frac{E(|Y_i|\; |\; |Y_i|>c_t)\Prob(|Y_i|>c_t)}{n_i}
\label{eq:mre}
\end{equation}
where $n_i$ is the true count, $|Y_i|$ is $|n_i-y_i|$, $y_i$ being the noisy count, and $\Prob(|Y_i|>c_t)=\#(|Y_i|>c_t)/\#Y_i$, where $\#(\cdot)$ represents the number of elements.

The third metric is $\Prob(|Y_i|\leq c_t)=\#(|Y_i|\leq c_t)/\#Y_i$, the probability that the error is within the breakpoint.

\subsubsection{Simulation Set-up}

We generate neighboring query outputs, $\ve{n_1},\ve{n_2}$ with original counts of $\ve{n}_1=[1,3,10,50,200,1000]$ and neighboring counts $\ve{n}_2=[1,3,10,50,200,1000]+1$ respectively. Hence the $\ell_1$ sensitivity $\Delta f =1$. We generate 120 million noise samples, $\ve{Y}$ $\in \Z$ independently from Laplace mixture (generated as specified in Algorithm 1), standard Laplace, where the noise samples are rounded to the nearest integer, and from geometric and geometric mixture (according to Algorithm 2). Then 10 million noise samples, for each mechanism, are added to each of the six original and six neighboring counts to generate differentially private output. For the cases where $n+Y<0$ we set $Y=-n$ (thus, for instance, for differentially private output zero counts occur very regularly for true counts of 1, and 2). According to Table I, in the Appendix, for the mixture mechanisms we choose two sets of values of break-point, $\e$ and $r\e$, $\{c_t,\e,r\e\}=\{5,1/5,1\}$ and $\{c_t,\e,r\e\}=\{6,1/10,1\}$.  The standard Laplace and geometric mechanisms are simulated such that their $\e=\zeta_{\e}$ of the Laplace mixture and geometric mixture mechanisms. Thus these are set at 0.328 when the mixture mechanism $\{c_t,\e,r\e\}=\{5,1/5,1\}$ and 0.257 respectively when $\{c_t,\e,r\e\}=\{6,1/10,1\}$ for the mixture mechanism. We also compare with more relaxed privacy budgets of 0.5 and 0.45 for standard geometric and Laplace mechanisms.



%

%

\subsubsection{Simulation Results}

We show the probability that the absolute error for $\{c_t,\e,r\e\}=\{5,1/5,1\}$ is within chosen break-point, as a bound, of $c_t=5$; for Laplace mixture and standard Laplace in Fig. \ref{fig:Bounds1} , and geometric mixture and standard geometric in Fig. \ref{fig:Bounds2}, with respect to true counts $\in \ve{n_1}$. There is 5\% improvement in error for the mixture mechanisms with small true counts less than 10, with 95\% within chosen $c_t=5$ bound, respect to equivalent privacy budget of standard mechanisms, with equivalent performance to the less private $\e=1/2$ cases for the standard Laplace and geometric mechanisms in Figs. \ref{fig:Bounds1} and \ref{fig:Bounds2} respectively. For true counts of 10 and above there is a 10\% improvement of true counts with respect to the equivalent privacy budget, with 92\% of counts within bound for both mixture mechanisms in Figs. \ref{fig:Bounds1} and \ref{fig:Bounds2}, with equivalent performance to the less private $\e=1/2$ cases.

\begin{figure}[t!]
\centering
\subfloat[Laplace mixture]{\label{fig:Bounds1} \includegraphics[width=0.7\columnwidth]{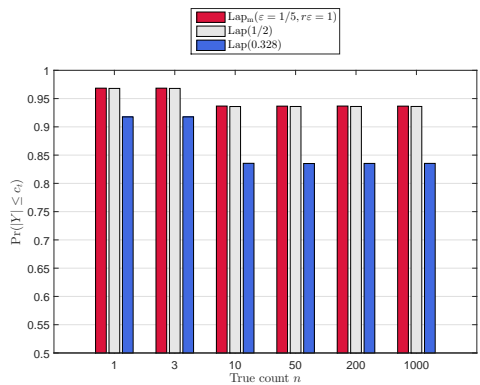}}\\
\subfloat[Geometric mixture]{\label{fig:Bounds2}\includegraphics[width=0.7\columnwidth]{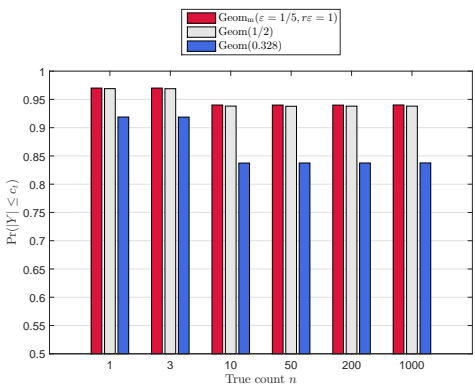}}
\caption{Probability noise is within bounds $c_t$, $\#(|Y_i|\leq c_t)/\#Y_i$ Laplace and geometric mixtures, $\e=1/5$, $r\e=1, c_t=5$ and standard Laplace and geometric mechanisms with $\e=0.328$ and $\e=1/2$}\vspace{-1mm}
\end{figure}

In Fig. \ref{fig:Err1} we plot the CDF of accuracy for $\{c_t,\e,r\e\}=\{5,1/5,1\}$ for the mixture mechanisms, where we note that there is equal performance accuracy of the geometric and Laplace mixture mechanisms, approximating the accuracy of the relaxed privacy budget of standard mechanisms, $\e=1/2$, within the break-point of 5, with $95\%$ of the added noise being within the break-point, as opposed to $85\%$ for the equivalent privacy budget. Above the break-point, the mixture mechanisms approach $100\%$ accuracy, more rapidly than even the relaxed privacy budget, within absolute error bounds of 8, as opposed to 12 for the relaxed privacy budget $\e=1/2$, and 16 for the equivalent privacy budget. In Fig. \ref{fig:Err2} we plot the CDF of accuracy for $\{c_t,\e,r\e\}=\{6,1/10,1\}$ for the mixture mechanisms, close to the 96\% accuracy of the relaxed privacy budget standard mechanisms, $\e=0.45$, within the break-point of 6. Then $94\%$ of the added noise is within the break-point, as opposed to $84\%$ for the equivalent privacy budget. Above the break-point, the mixture mechanisms approach $100\%$ accuracy, more rapidly than even the relaxed privacy budget, within absolute error bounds of 10, as opposed to 11 for the relaxed privacy budget $\e=1/2$, and 21 for equivalent privacy budget.

\begin{figure}[t!]
\centering
\subfloat[Geometric and Laplace mixtures, $\e=1/5$, $r\e=1, c_t=5$; standard geometric and Laplace mechanisms with $\e=0.328$ and $\e=1/2$]{\label{fig:Err1}\includegraphics[width=0.7\columnwidth]{hat1_mat_figs/Utility_CDF_newn.eps}\vspace{-2mm}}\\
\subfloat[Geometric and Laplace mixtures, $\e=1/10$, $r\e=1, c_t=6$; standard geometric and Laplace mechanisms with $\e=0.257$ and $\e=0.45$]{\label{fig:Err2}\includegraphics[width=0.7\columnwidth]{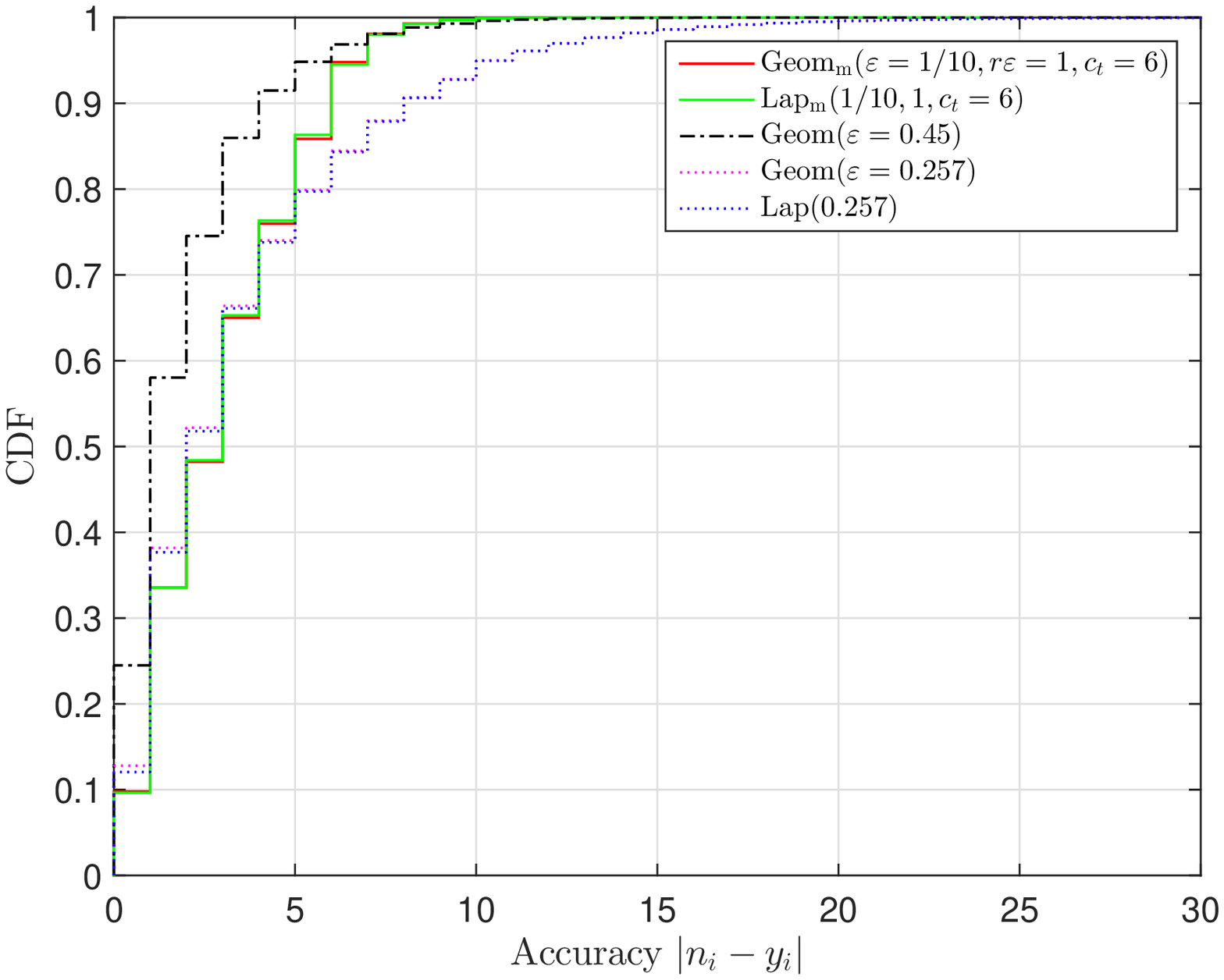}}
\caption{CDF of Accuracy $|Count_{True}-Count_{Noisy}|=|n_i-y_i|$ }\vspace{-4mm}
\end{figure}

In terms of weighted mean relative error for $|Y_i| > c_t$, according to (\ref{eq:mre}), with $\{c_t,\e,r\e\}=\{5,1/5,1\}$, as shown in Fig. \ref{fig:Meanbigerr1}, there is a factor of 3 improvement for the mixture mechanisms with respect to equivalent privacy budget of standard mechanisms, with slightly improved performance, by a factor of 1.05, to the less private $\e=1/2$ cases for the standard geometric mechanisms (which is also equivalent to that for the same privacy budget for standard Laplace mechanism). Furthermore it is clear that there is very small relative error in Fig. \ref{fig:Meanbigerr1} for the mixture mechanisms for any counts larger than 10, error less than 0.01, and small relative error for counts above 1. In terms of mean relative error for $\{c_t,\e,r\e\}=\{6,1/10,1\}$ as shown in Fig. \ref{fig:Meanbigerr2} there is a factor of 4 improvement for the mixture mechanisms with respect to equivalent privacy budget of 0.257 of standard mechanisms. There is slightly deteriorated performance (by a factor of 1.5) to the less private $\e=0.45$ cases for the standard geometric mechanism. Furthermore it is clear that there is very small relative error in Fig. \ref{fig:Meanbigerr1} for the mixture mechanisms for any counts larger than 10, error less than 0.01, and small relative error for counts above 1. For $\{c_t,\e,r\e\}=\{6,1/10,1\}$ as shown in Fig. \ref{fig:Meanbigerr2}, in comparison to $\{c_t,\e,r\e\}=\{5,1/5,1\}$ as shown in Fig. \ref{fig:Meanbigerr1}, there is an increase in relative error, for this improved privacy budget, only by a factor of 1.2 across all true counts $n$, even though the general privacy budget $\zeta_{\e}$ has improved from 0.328 to 0.257.

\begin{figure}[t!]
\centering
\subfloat[Geometric and Laplace mixtures, $\e=1/5$, $r\e=1, c_t=5$; standard geometric and Laplace mechanisms with $\e=0.328$ and $\e=1/2$]
{\label{fig:Meanbigerr1}\includegraphics[width=0.7\columnwidth]{hat1_mat_figs/Mean_rel_bigerrn.eps}\vspace{-2mm}}\\
\subfloat[Geometric and Laplace mixtures, $\e=1/10$, $r\e=1, c_t=6$; standard geometric and Laplace mechanisms with $\e=0.257$ and $\e=0.45$]
{\label{fig:Meanbigerr2}\includegraphics[width=0.7\columnwidth]{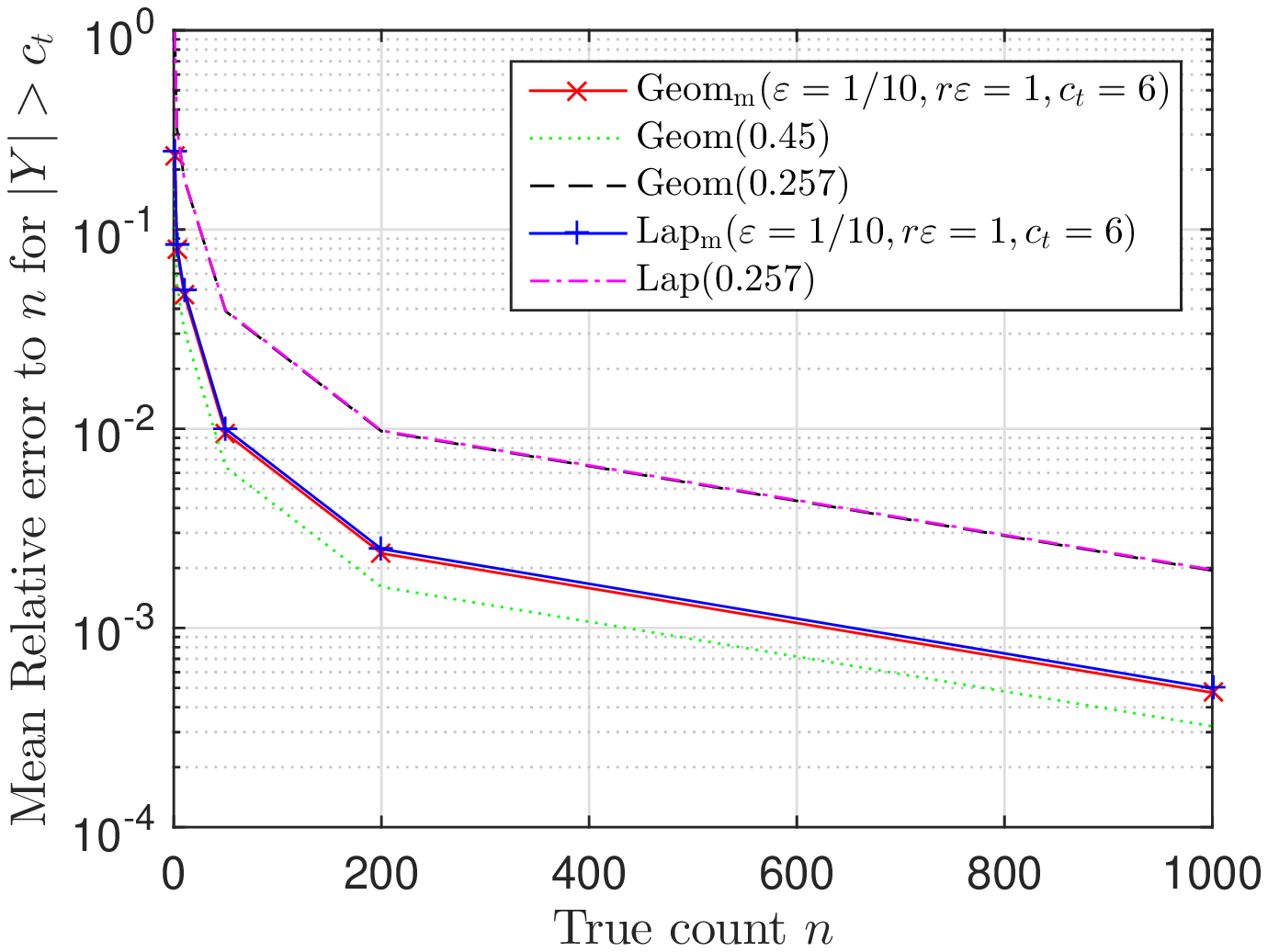}}
\caption{Mean relative error fraction $\frac{E(|Y_i|\;|\;|Y_i|>c_t)\Prob(|Y_i|>c_t)}{n_i}$}\vspace{-2mm}
\end{figure}

%

\vspace{-3pt}
\subsection{Performance Evaluation with a Real-World Dataset}
We use the Adult dataset from the UCI Machine Learning Repository\footnote{\url{https://archive.ics.uci.edu/ml/datasets/Adult}}, extracted from 1994 US census data, which has been widely used for differential-privacy benchmarking, including recently in, e.g., \cite{zhu2015,boyd2015,gaboardi2014}. This dataset contains 32,651 unit records of Census data with 14 attributes. The dataset manifests 8 categorical attributes and 4 continuous integer attributes. For the evaluation here, we focus only on the perturbation of count queries with a random combination of two attribute values. Therefore, we first study the distribution of true counts and the characteristics of neighboring databases, i.e., two databases where one database contains records of one additional user, from the Adult dataset.

\begin{figure}[tb]
\centering
\subfloat[Distribution of random query counts]{\label{fig:random_query}\includegraphics[width=0.75\columnwidth]{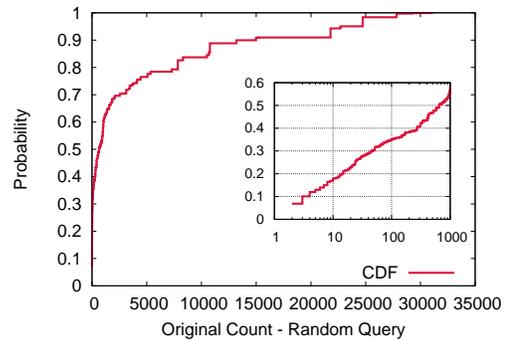}}\\\vspace{-2mm}
\subfloat[Impact of each unit record]{\label{fig:ndb_user}\includegraphics[width=0.75\columnwidth]{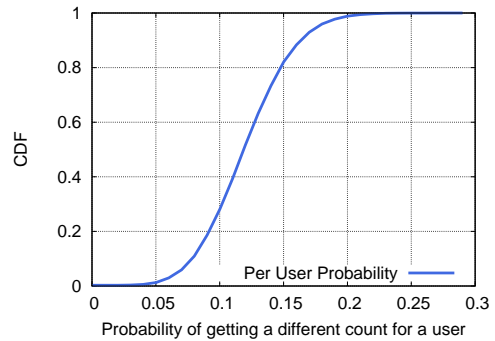}}\\
\caption{Characteristics of Adult dataset}\vspace{-4mm}
\label{fig:adult}
\end{figure}

\subsubsection{Characteristics of the dataset in use}

Fig. \ref{fig:random_query} illustrates the distribution of true counts for 5 million random combinations of two attribute values. Although such true counts exhibit a wide range, approximately 20\% of the queries result in a count less than 10. The probability of the query result being a low count is a significant characteristic of a dataset, as it increases the general privacy loss due to the rounding of negative perturbed counts to zero. However, the majority of the query results are larger values compared to the amount of noise added from the proposed mechanisms. As a result, utility measures for the Adult dataset will be higher due to the smaller error relative to the true count.

\begin{figure*}[t]
\centering
\subfloat[Absolute error]{\label{fig:abs_error}\includegraphics[width=0.7\columnwidth]{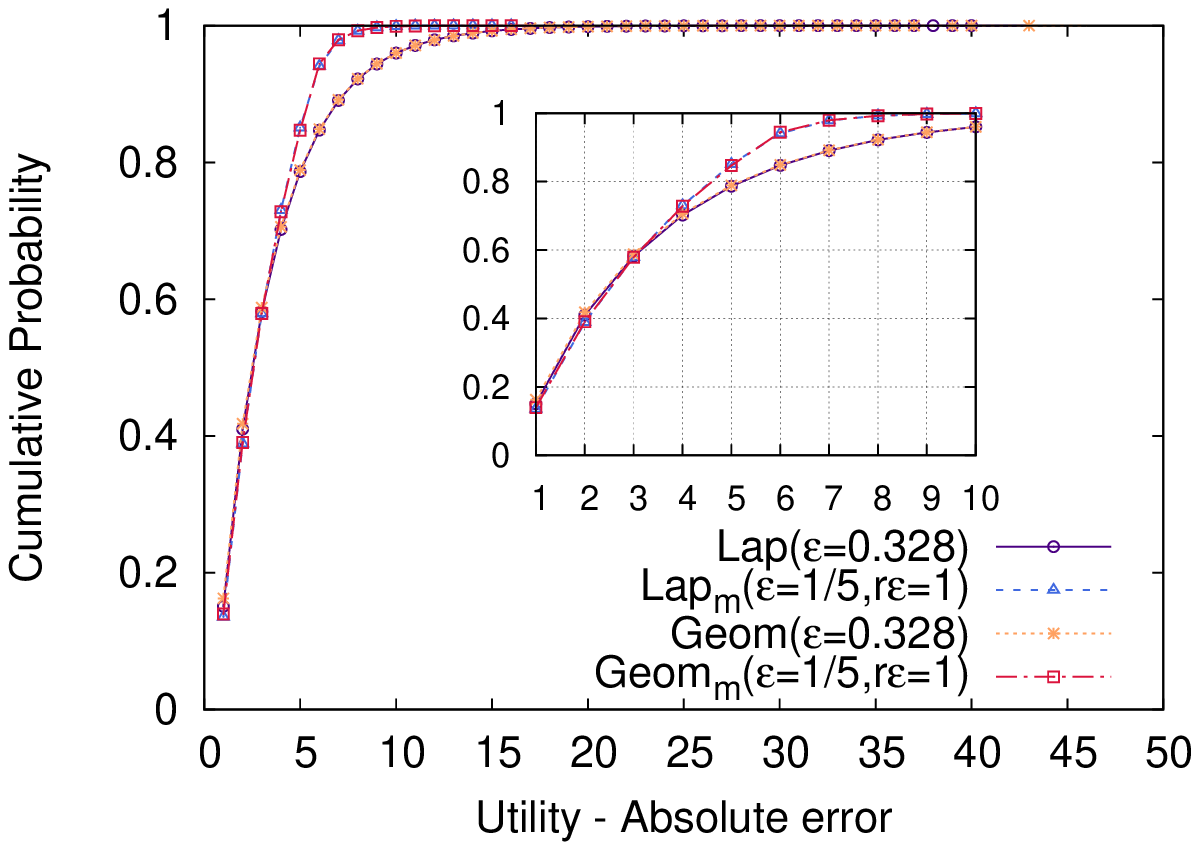}}
\subfloat[CDF of mean relative error]{\label{fig:truncation_error}\includegraphics[width=0.7\columnwidth]{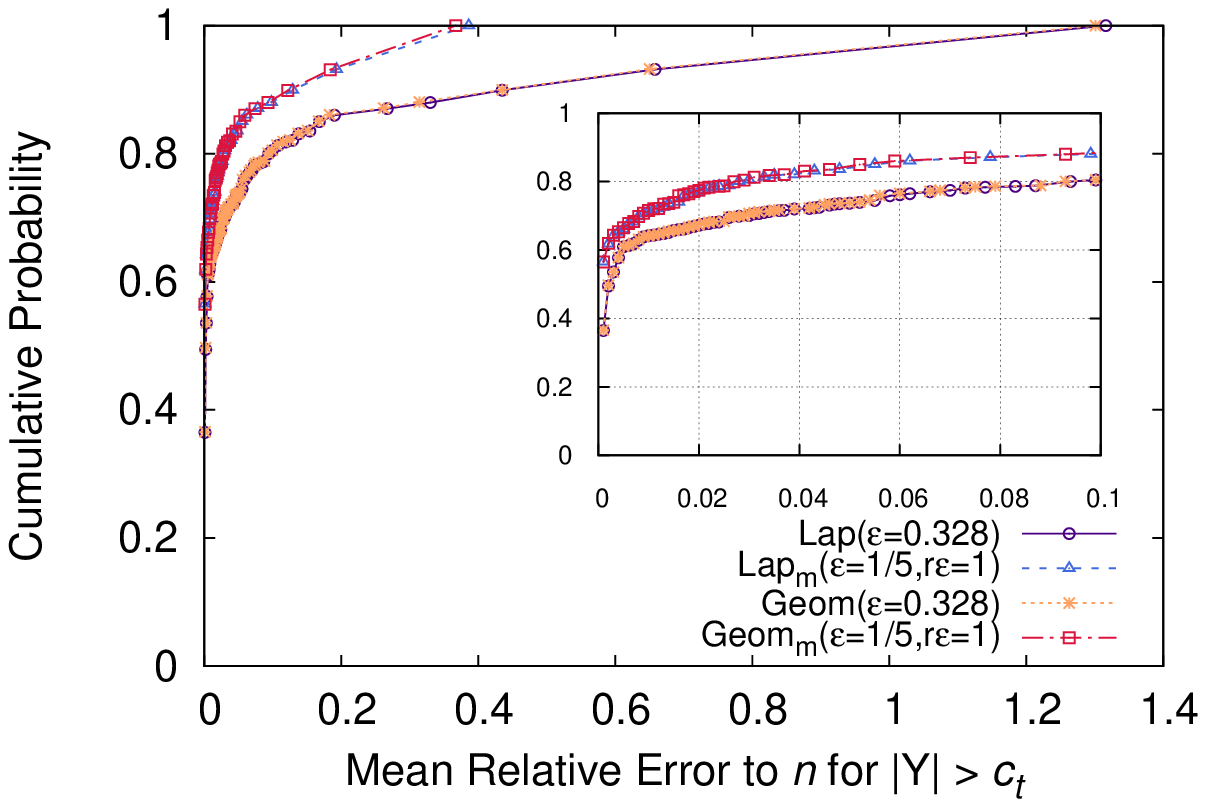}}
\subfloat[Mean relative error vs true count]{\label{fig:error_tc}\includegraphics[width=0.7\columnwidth]{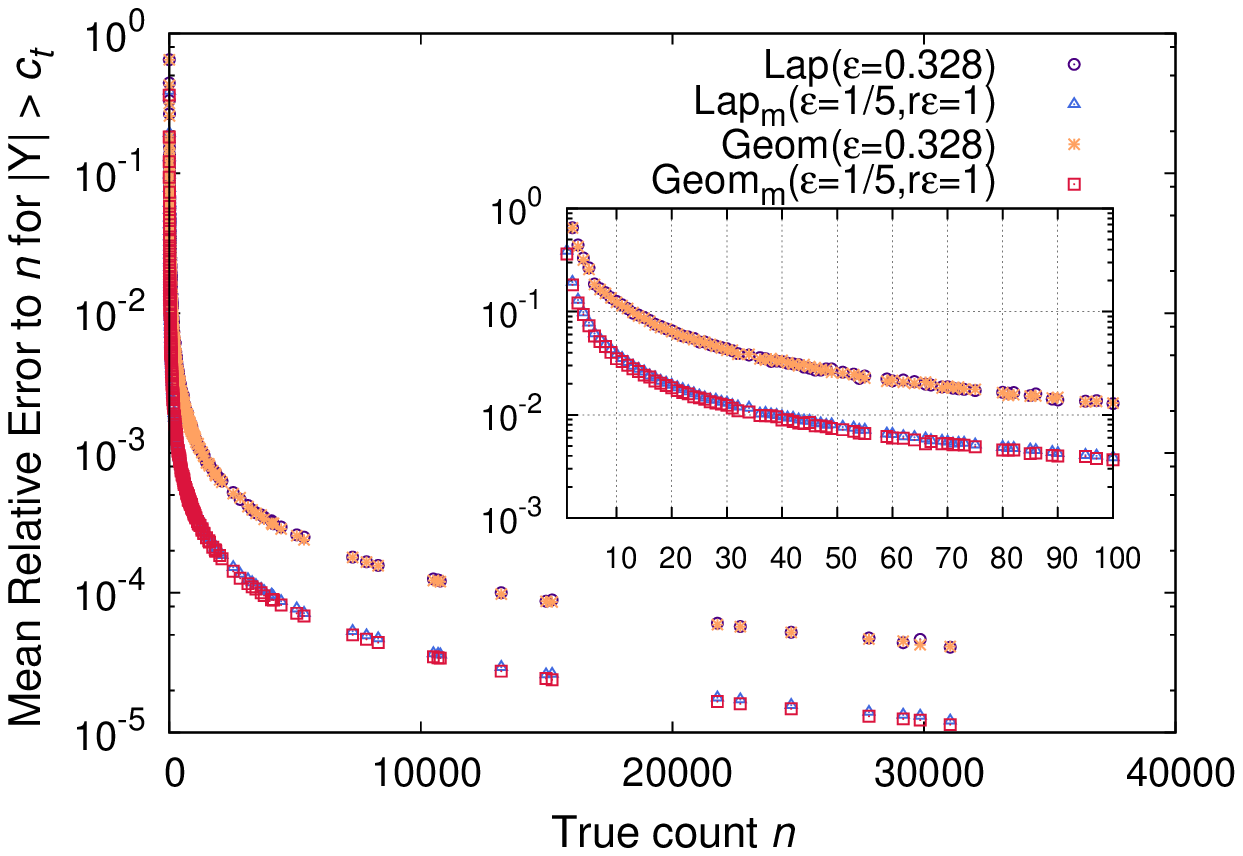}}
\caption{Utility measures for 1 million random queries on Adult dataset. Absolute error $|Count_{True}-Count_{Noisy}|$ and mean relative error fraction $\frac{E(|Y_i|\;|\;|Y_i|>c_t)\Prob(|Y_i|>c_t)}{n_i}$ geometric and Laplace mixtures, $\e=1/5$, $r\e=1, c_t=5$ and standard geometric and Laplace mechanisms with $\e=0.328$}
\label{fig:utility}
\end{figure*}

As we quantify the privacy loss by comparing two neighboring databases, the true difference in neighboring databases is also an important characteristic of a dataset. By removing each user from the dataset, we have created 32,650 different neighboring databases and then performed 100 random count queries from each database. We observe that approximately 85\% of query answers do not change for the two neighboring databases. Moreover, the probability of getting a different count for an individual user in any given two neighboring databases is also very low. Fig. \ref{fig:ndb_user} show that it is almost normally distributed with a mean value of 0.125. Fig. \ref{fig:ndb_user} further emphasizes the fact that probability of getting a different count is less than 0.2. Therefore, whether a person is in the database or not, is not revealed for approximately 80\% of the random queries even without a privacy preserving mechanism.

\subsubsection{Utility-Privacy analysis}

For utility-privacy analysis, we consider the performance analysis metrics, absolute error ($|Count_{True}-Count_{Noisy}|=|n_i -y_i|$) and mean relative error fraction $\frac{E(|Y_i|>c_t)\Prob(|Y_i|>c_t)}{n_i}$ defined in Section \ref{sec:metrics} and we compare the performance of the proposed \emph{piecewise mixture mechanisms} with equivalent standard Laplace and geometric mechanisms. The standard Laplace and geometric mechanisms are parameterized such that $\e = \zeta_\e$ of the mixture mechanisms, which resulted in $\e=0.328$ for standard mechanisms when the mixture mechanism $\{c_t, \e, r\e\} = \{5, 1/5, 1\}$.\\

\noindent{\bf Utility measures:}
The cumulative probability of absolute error ($|Count_{True}-Count_{Noisy}|$) for all four mechanisms are shown in Fig. \ref{fig:abs_error}. It shows that absolute error for \emph{piecewise mixture mechanisms} are less than 10 in almost in every case while the maximum also limited to 15 whereas absolute error for standard mechanisms spreads up to 40. The probability of absolute error being less than 4 is approximately similar for all mechanisms. This validates the design goals of \emph{piecewise mixture mechanisms} to reduce the probability of getting a larger error while slightly increasing the probability of getting a smaller error.

Fig. \ref{fig:truncation_error} further validates these aspects with the metric that is better designed to capture the impact of the break-point, mean relative error fraction $\frac{E(|Y_i|>c_t)\Prob(|Y_i|>c_t)}{n_i}$ (cf. ~\ref{eq:mre}). The lower the mean relative  error, the better the utility. \emph{Piecewise mixture mechanisms} provide lower relative error compared to standard mechanisms which are parameterised to provide the same privacy loss. It also shows that the \emph{geometric  piecewise mixture} provides slightly better utility as expected due to the discrete nature of the geometric mechanism and rounding errors of the Laplace mechanism. For larger $\e$ values, the performance difference between Laplace and geometric mechanisms increases as shown in Table~I, in the Appendix. To compare the experimental results with simulations, we plot the mean relative error against true count in Fig. ~\ref{fig:error_tc}. \emph{Piecewise mixture mechanisms} always result in less error irrespective of the true count. Furthermore, the results show similar patterns and almost similar error values compared to simulations in Fig.~\ref{fig:Meanbigerr2}.\\

\begin{figure}[t]
\centering
\subfloat[Similar neighboring databases]{\label{fig:priv_loss_same}\includegraphics[width=0.72\columnwidth]{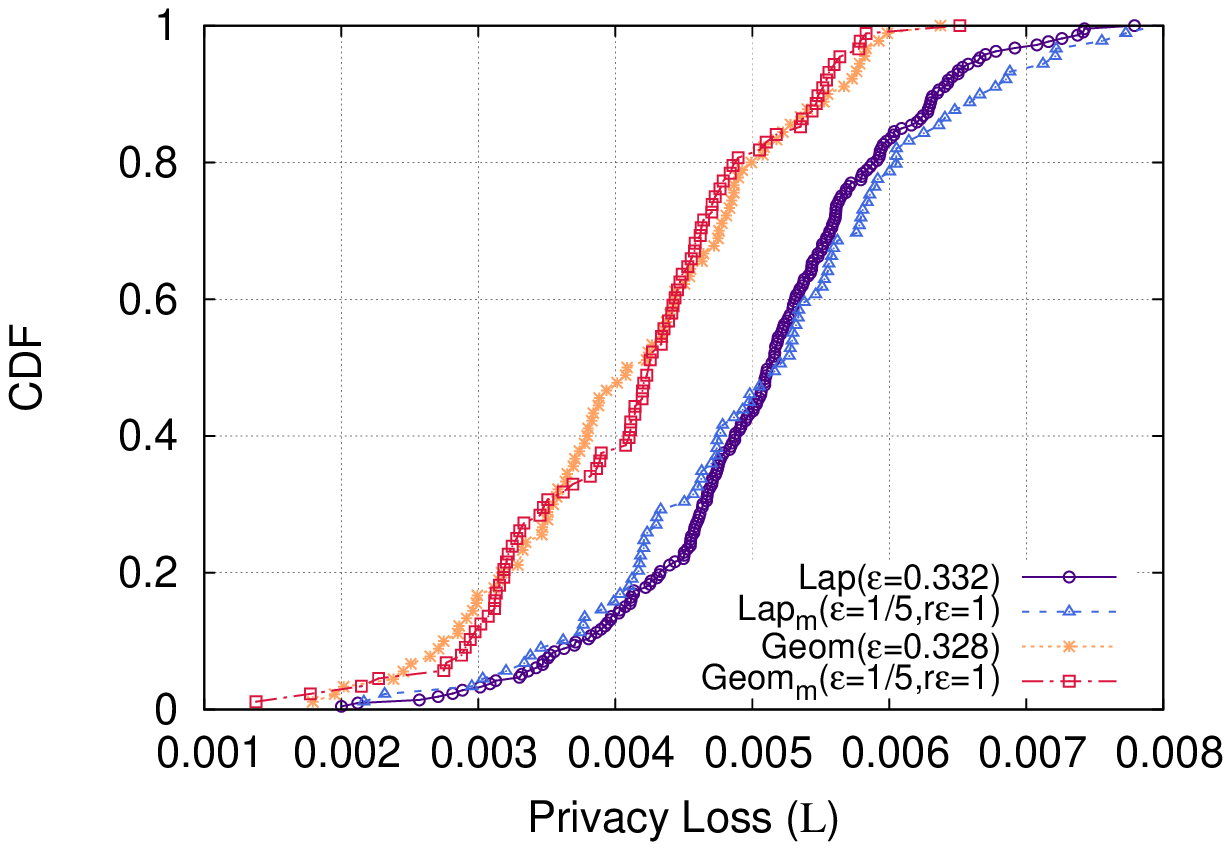}}\\
\subfloat[Different neighboring databases]{\label{fig:priv_loss_diff}\includegraphics[width=0.72\columnwidth]{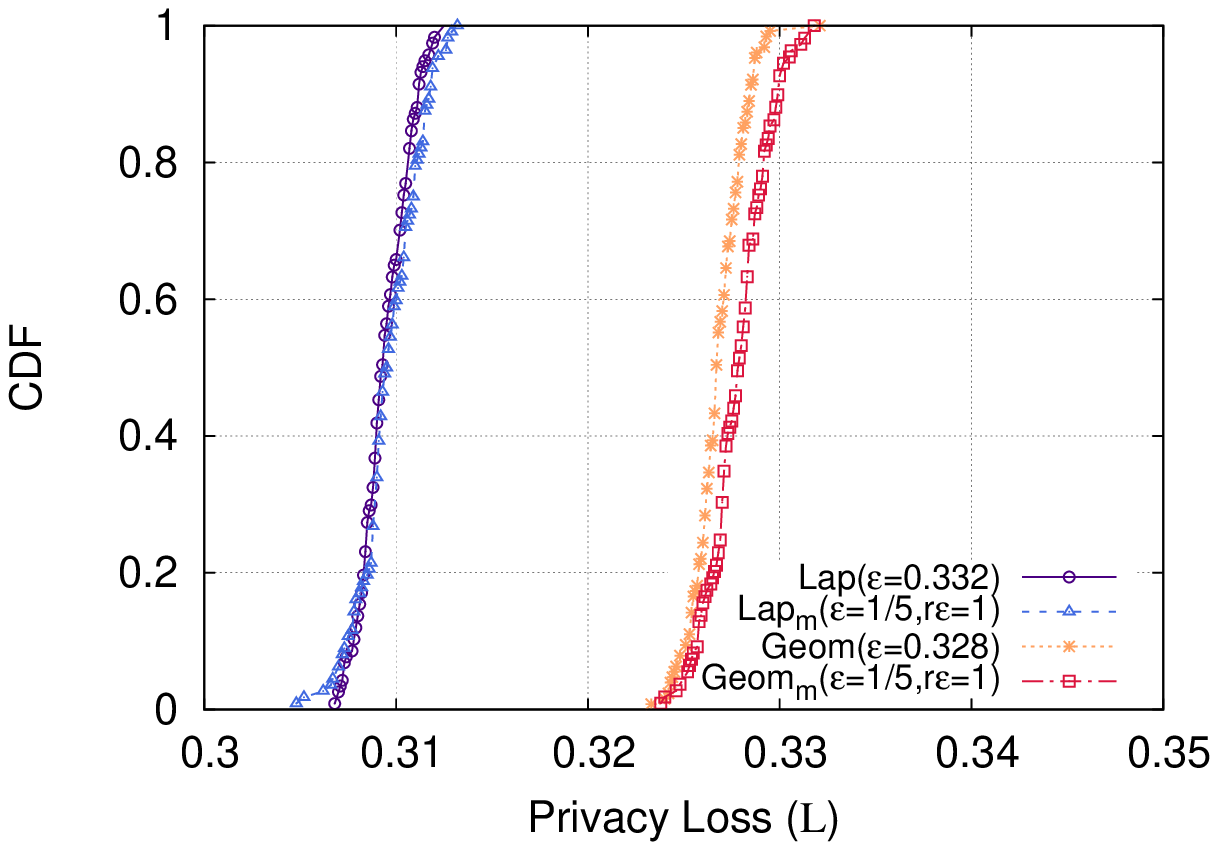}}
\caption{Privacy loss for neighboring databases. Privacy Loss $\L$, geometric and Laplace mixtures, $\e=1/5$, $r\e=1, c_t=5$ and standard geometric and Laplace mechanisms with $\e=0.328$ and $\e=0332$ respectively.}\vspace{-6mm}
\label{fig:privacy}
\end{figure}
\vspace{-5pt}
\noindent{\bf Privacy measures:}
Since differential privacy is based on the premise that there should not be any additional privacy risk when an individual is in or not in the database, we quantify the privacy loss of individuals in the Adult dataset by iteratively removing every user from the dataset. When the neighboring databases are providing the same answer to a query, the privacy loss is almost zero irrespective of the private mechanism used. As the majority (80\%) of neighboring databases are in this category (cf. Fig. \ref{fig:ndb_user}), average privacy loss value is dominated by the similar neighboring databases. Therefore, we analyse the privacy loss $\L$, according to (\ref{ploss}), separately for similar and different neighboring databases in Fig.~\ref{fig:priv_loss_same} and~\ref{fig:priv_loss_diff} respectively. To compensate for rounding errors of continuous Laplace mechanism, we use $\e=0.332$ for standard Laplace mechanism, instead of $\e=0.328$ for standard geometric mechanism.

Fig.~\ref{fig:priv_loss_same} depicts that privacy loss is almost zero, ($<$ 0.008), for all queries when the two  neighboring databases result in the same count. There is a jump in the privacy loss value up to the region closer to the used privacy budget when the two  neighboring databases result in the different counts as shown in Fig.~\ref{fig:priv_loss_diff}.
Please note that in both cases, \emph{piecewise mixture mechanisms} follow the privacy loss of their equivalent standard mechanisms, whilst the utility of \emph{piecewise mixture mechanisms} are superior to standard mechanisms, as shown in Fig. \ref{fig:utility}. Thus, these results experimentally validate the design goals of \emph{piecewise mixture mechanisms} as they provide better utility for the same level of privacy loss.

\section{Conclusion}

In this paper, we presented a novel approach to differentially private statistical distribution mechanisms, giving greater flexibility to database curators, and providing more accuracy to analysts. This has been achieved by deriving piecewise mixture mechanisms, building from the classical Laplace and symmetric geometric mechanisms that have found wide application for ensuring differentially private query release. In terms of a newly defined parameter, general privacy budget, closely related to $\e$ in standard differential privacy, the Laplace and geometric mixture mechanisms demonstrate better performance in terms of various metrics for loss, entropy as well as accuracy and maintaining added noise within suitable bounds. Moreover the piecewise mixture distributions enable mechanism design that can approximate those from truncated distributions, designed for better utility, without sacrificing differential privacy, which may often occur if truncated distributions are applied. Importantly, the properties of the Laplace and geometric piecewise mixture mechanisms are preserved under composition, and are very advantageous for iterative online dataset querying such as that in the classical online private multiplicative weights algorithm, enabling more querying, requiring less dataset updates, and better accuracy in released data. Theoretical analysis, simulation and empirical testing on an open-access dataset has confirmed the favorable properties of the Laplace and geometric piecewise mechanisms, mitigating loss, reducing entropy providing greater accuracy with respect to general privacy, and enabling most noise to be added within close numeric bounds to the true data.

In future work we will be applying the piecewise mixture mechanisms in an online linear iterative setting, determining appropriate thresholds for dataset updates and precise quantification of the increase in number of possible query answers, employing a mixture mechanism, while maintaining differential privacy. In other future work we will seek to evaluate the performance of piecewise mixture mechanisms in the more-relaxed approximate differential privacy setting.



\begin{table*}[h]
\centering
\normalsize APPENDIX
\vspace{8mm}
\scriptsize
\caption{\footnotesize Metrics of loss, $\chi = E(|x|), \chi = \sigma^2(x)$ and entropy $H(x)$ , for standard Laplace, geometric, and piecewise mixtures, for various $c_t$, $r\e$ and $\e$, general privacy budgets $\zeta_\e$, equivalent $\e$ for rounded standard Laplace and for standard geometric $\e =\zeta_\e\{\Gm\}$ \vspace{-2mm}}
\centering
\begin{tabular}{|c|c|c|c|c|c|c|}
\hline
\multirow{2}{*}{${\scriptstyle c_t}$}&\multirow{2}{*}{${\scriptstyle \e}$}&\multirow{2}{*}{${\scriptstyle r\e}$}&\multirow{2}{*}{$\scriptscriptstyle \begin{array}{c}\scriptstyle\zeta_\e\{\Gm,\Lm\} \\ \scriptstyle{\e}_{\Lap}\end{array}$}&
\multirow{2}{*}{$\begin{array}{c}\scriptstyle E(|x|)[\{\Gm,\Lm\}\\ \scriptstyle \{\Geo,\Lap\}]\end{array}$}&\multirow{2}{*}{$\begin{array}{c}\scriptstyle \sigma^2[\{\Gm,\Lm\}\\ \scriptstyle \{\Geo,\Lap\}]\end{array}$}&\multirow{2}{*}{$\begin{array}{c}\scriptstyle H(x)[\{\Gm,\Lm\},\\ \scriptstyle \{\Geo,\Lap\}]\end{array}$}\\
& & & & & &\\
\whline
4&0.1&0.2&\{0.152,0.148\}, 0.154&\{5.44,5.46\}, \{6.57,6.50\}&\{55.72,55.91\},  \{86.71,84.86\}&\{3.38,3.39\},  \{3.58,3.57\}\\\hline
4&0.1&0.4&\{0.212,0.207\}, 0.218&\{3.41,3.42\}, \{4.67,4.59\}&\{19.49,19.62\},  \{44.20,42.32\}&\{2.89,2.89\},  \{3.24,3.22\}\\\hline
4&0.1&0.5&\{0.235,0.228\}, 0.241&\{3.04,3.05\}, \{4.22,4.14\}&\{14.96,15.07\},  \{36.15,34.50\}&\{2.76,2.76\},  \{3.14,3.12\}\\\hline
4&0.1&1&\{0.334,0.316\}, 0.339&\{2.39,2.39\}, \{2.94,2.93\}&\{8.47,8.48\},  \{17.82,17.45\}&\{2.46,2.46\},  \{2.78,2.77\}\\\hline
4&0.167&0.333&\{0.228,0.219\}, 0.231&\{3.55,3.57\}, \{4.35,4.32\}&\{22.88,23.04\},  \{38.41,37.55\}&\{2.96,2.96\},  \{3.17,3.16\}\\\hline
4&0.167&0.667&\{0.297,0.285\}, 0.304&\{2.53,2.55\}, \{3.32,3.28\}&\{10.32,10.41\},  \{22.53,21.72\}&\{2.58,2.58\},  \{2.90,2.88\}\\\hline
4&0.167&0.833&\{0.326,0.310\}, 0.333&\{2.36,2.36\}, \{3.02,2.99\}&\{8.68,8.74\},  \{18.67,18.10\}&\{2.49,2.49\},  \{2.81,2.79\}\\\hline
4&0.167&1.67&\{0.501,0.445\}, 0.491&\{2.05,2.04\}, \{1.92,2.02\}&\{6.27,6.19\},  \{7.81,8.38\}&\{2.30,2.29\},  \{2.36,2.40\}\\\hline
4&0.2&0.4&\{0.263,0.251\}, 0.267&\{3.08,3.11\}, \{3.76,3.74\}&\{17.04,17.20\},  \{28.84,28.23\}&\{2.81,2.82\},  \{3.02,3.02\}\\\hline
4&0.2&0.8&\{0.335,0.319\}, 0.343&\{2.31,2.32\}, \{2.93,2.90\}&\{8.52,8.60\},  \{17.62,17.11\}&\{2.48,2.48\},  \{2.78,2.76\}\\\hline
4&0.2&1&\{0.368,0.347\}, 0.375&\{2.17,2.18\}, \{2.65,2.65\}&\{7.39,7.44\},  \{14.57,14.28\}&\{2.41,2.41\},  \{2.68,2.67\}\\\hline
4&0.2&2&\{0.597,0.512\}, 0.572&\{1.95,1.93\}, \{1.58,1.72\}&\{5.73,5.62\},  \{5.45,6.19\}&\{2.25,2.24\},  \{2.18,2.25\}\\\hline
4&0.25&0.5&\{0.313,0.296\}, 0.317&\{2.61,2.64\}, \{3.15,3.14\}&\{12.10,12.25\},  \{20.28,19.93\}&\{2.65,2.65\},  \{2.85,2.84\}\\\hline
4&0.25&1&\{0.391,0.366\}, 0.398&\{2.07,2.08\}, \{2.49,2.50\}&\{6.88,6.94\},  \{12.93,12.71\}&\{2.37,2.38\},  \{2.62,2.61\}\\\hline
4&0.25&1.25&\{0.431,0.399\}, 0.436&\{1.98,1.98\}, \{2.25,2.28\}&\{6.17,6.20\},  \{10.60,10.61\}&\{2.31,2.31\},  \{2.52,2.52\}\\\hline
4&0.25&2.5&\{0.761,0.619\}, 0.706&\{1.83,1.81\}, \{1.20,1.39\}&\{5.15,5.00\},  \{3.29,4.09\}&\{2.19,2.18\},  \{1.92,2.04\}\\\hline
4&0.5&1&\{0.548,0.497\}, 0.554&\{1.57,1.62\}, \{1.73,1.78\}&\{4.59,4.74\},  \{6.48,6.59\}&\{2.16,2.17\},  \{2.27,2.28\}\\\hline
4&0.5&2&\{0.654,0.576\}, 0.652&\{1.44,1.47\}, \{1.42,1.51\}&\{3.68,3.74\},  \{4.51,4.79\}&\{2.05,2.05\},  \{2.08,2.12\}\\\hline
4&0.5&2.5&\{0.744,0.633\}, 0.723&\{1.42,1.44\}, \{1.23,1.35\}&\{3.56,3.58\},  \{3.45,3.91\}&\{2.03,2.03\},  \{1.95,2.02\}\\\hline
5&0.1&0.2&\{0.145,0.142\}, 0.146&\{5.63,5.64\}, \{6.88,6.82\}&\{58.44,58.63\},  \{95.19,93.37\}&\{3.42,3.42\},  \{3.62,3.61\}\\\hline
5&0.1&0.4&\{0.194,0.189\}, 0.198&\{3.72,3.73\}, \{5.13,5.05\}&\{22.64,22.75\},  \{53.16,51.34\}&\{2.97,2.97\},  \{3.33,3.31\}\\\hline
5&0.1&0.5&\{0.211,0.206\}, 0.216&\{3.39,3.39\}, \{4.70,4.63\}&\{18.08,18.17\},  \{44.66,43.05\}&\{2.86,2.86\},  \{3.24,3.23\}\\\hline
5&0.1&1&\{0.289,0.274\}, 0.292&\{2.79,2.78\}, \{3.41,3.42\}&\{11.40,11.36\},  \{23.72,23.60\}&\{2.61,2.61\},  \{2.93,2.93\}\\\hline
5&0.167&0.333&\{0.216,0.208\}, 0.219&\{3.75,3.77\}, \{4.59,4.57\}&\{25.01,25.17\},  \{42.69,41.93\}&\{3.01,3.01\},  \{3.22,3.21\}\\\hline
5&0.167&0.667&\{0.269,0.258\}, 0.274&\{2.84,2.85\}, \{3.68,3.64\}&\{12.77,12.85\},  \{27.51,26.81\}&\{2.69,2.69\},  \{3.00,2.99\}\\\hline
5&0.167&0.833&\{0.291,0.277\}, 0.295&\{2.68,2.69\}, \{3.39,3.37\}&\{11.14,11.18\},  \{23.47,23.00\}&\{2.61,2.61\},  \{2.92,2.91\}\\\hline
5&0.167&1.67&\{0.428,0.380\}, 0.414&\{2.41,2.39\}, \{2.27,2.40\}&\{8.66,8.54\},  \{10.77,11.73\}&\{2.46,2.45\},  \{2.53,2.57\}\\\hline
5&0.2&0.4&\{0.249,0.238\}, 0.252&\{3.28,3.30\}, \{3.97,3.96\}&\{18.94,19.09\},  \{32.07,31.58\}&\{2.87,2.87\},  \{3.08,3.07\}\\\hline
5&0.2&0.8&\{0.303,0.288\}, 0.308&\{2.60,2.61\}, \{3.25,3.23\}&\{10.73,10.79\},  \{21.57,21.17\}&\{2.60,2.60\},  \{2.88,2.87\}\\\hline
5&0.2&1&\{0.328,0.309\}, 0.332&\{2.48,2.49\}, \{2.99,3.00\}&\{9.61,9.63\},  \{18.41,18.25\}&\{2.54,2.54\},  \{2.80,2.80\}\\\hline
5&0.2&2&\{0.506,0.435\}, 0.479&\{2.28,2.27\}, \{1.90,2.07\}&\{7.92,7.77\},  \{7.66,8.81\}&\{2.41,2.41\},  \{2.35,2.43\}\\\hline
5&0.25&0.5&\{0.297,0.281\}, 0.300&\{2.79,2.82\}, \{3.32,3.32\}&\{13.71,13.86\},  \{22.51,22.28\}&\{2.71,2.72\},  \{2.90,2.90\}\\\hline
5&0.25&1&\{0.353,0.331\}, 0.357&\{2.33,2.35\}, \{2.77,2.78\}&\{8.77,8.82\},  \{15.85,15.73\}&\{2.49,2.50\},  \{2.72,2.72\}\\\hline
5&0.25&1.25&\{0.383,0.355\}, 0.384&\{2.25,2.26\}, \{2.55,2.59\}&\{8.09,8.10\},  \{13.49,13.62\}&\{2.45,2.45\},  \{2.64,2.65\}\\\hline
5&0.25&2.5&\{0.637,0.520\}, 0.582&\{2.13,2.11\}, \{1.47,1.70\}&\{7.06,6.89\},  \{4.77,5.99\}&\{2.36,2.35\},  \{2.11,2.23\}\\\hline
5&0.5&1&\{0.529,0.479\}, 0.532&\{1.67,1.72\}, \{1.81,1.86\}&\{5.31,5.47\},  \{6.98,7.15\}&\{2.22,2.24\},  \{2.31,2.32\}\\\hline
5&0.5&2&\{0.593,0.526\}, 0.590&\{1.58,1.62\}, \{1.59,1.67\}&\{4.56,4.64\},  \{5.53,5.83\}&\{2.15,2.16\},  \{2.19,2.22\}\\\hline
5&0.5&2.5&\{0.649,0.561\}, 0.632&\{1.56,1.60\}, \{1.44,1.56\}&\{4.46,4.51\},  \{4.58,5.09\}&\{2.14,2.14\},  \{2.09,2.15\}\\\hline
6&0.1&0.2&\{0.139,0.136\}, 0.140&\{5.82,5.83\}, \{7.17,7.12\}&\{61.50,61.67\},  \{103.25,101.58\}&\{3.45,3.45\},  \{3.66,3.66\}\\\hline
6&0.1&0.4&\{0.179,0.175\}, 0.182&\{4.04,4.05\}, \{5.55,5.48\}&\{26.15,26.25\},  \{62.15,60.39\}&\{3.04,3.05\},  \{3.41,3.40\}\\\hline
6&0.1&0.5&\{0.193,0.188\}, 0.197&\{3.73,3.73\}, \{5.14,5.08\}&\{21.56,21.63\},  \{53.37,51.80\}&\{2.95,2.95\},  \{3.33,3.32\}\\\hline
6&0.1&1&\{0.257,0.243\}, 0.257&\{3.17,3.16\}, \{3.85,3.88\}&\{14.71,14.64\},  \{30.19,30.34\}&\{2.73,2.73\},  \{3.05,3.05\}\\\hline
6&0.167&0.333&\{0.207,0.199\}, 0.209&\{3.95,3.96\}, \{4.80,4.78\}&\{27.30,27.46\},  \{46.54,45.91\}&\{3.05,3.06\},  \{3.26,3.26\}\\\hline
6&0.167&0.667&\{0.248,0.238\}, 0.251&\{3.12,3.13\}, \{3.99,3.97\}&\{15.44,15.50\},  \{32.32,31.73\}&\{2.78,2.78\},  \{3.08,3.07\}\\\hline
6&0.167&0.833&\{0.265,0.253\}, 0.268&\{2.98,2.99\}, \{3.73,3.72\}&\{13.81,13.83\},  \{28.26,27.92\}&\{2.72,2.72\},  \{3.01,3.01\}\\\hline
6&0.167&1.67&\{0.374,0.334\}, 0.360&\{2.74,2.72\}, \{2.61,2.76\}&\{11.30,11.15\},  \{14.12,15.48\}&\{2.59,2.59\},  \{2.66,2.71\}\\\hline
6&0.2&0.4&\{0.239,0.228\}, 0.241&\{3.46,3.48\}, \{4.15,4.14\}&\{20.95,21.10\},  \{34.90,34.51\}&\{2.92,2.93\},  \{3.12,3.12\}\\\hline
6&0.2&0.8&\{0.280,0.266\}, 0.283&\{2.86,2.87\}, \{3.52,3.52\}&\{13.08,13.14\},  \{25.31,25.01\}&\{2.70,2.70\},  \{2.96,2.95\}\\\hline
6&0.2&1&\{0.299,0.282\}, 0.301&\{2.76,2.77\}, \{3.29,3.31\}&\{11.98,11.99\},  \{22.18,22.14\}&\{2.65,2.65\},  \{2.89,2.89\}\\\hline
6&0.2&2&\{0.439,0.379\}, 0.413&\{2.59,2.57\}, \{2.21,2.40\}&\{10.29,10.11\},  \{10.23,11.79\}&\{2.55,2.54\},  \{2.50,2.58\}\\\hline
6&0.25&0.5&\{0.286,0.270\}, 0.288&\{2.96,2.98\}, \{3.46,3.46\}&\{15.36,15.51\},  \{24.37,24.22\}&\{2.77,2.77\},  \{2.94,2.94\}\\\hline
6&0.25&1&\{0.327,0.307\}, 0.330&\{2.57,2.58\}, \{3.00,3.02\}&\{10.73,10.78\},  \{18.53,18.50\}&\{2.59,2.59\},  \{2.80,2.80\}\\\hline
6&0.25&1.25&\{0.349,0.324\}, 0.349&\{2.50,2.51\}, \{2.81,2.85\}&\{10.07,10.08\},  \{16.26,16.49\}&\{2.56,2.56\},  \{2.74,2.75\}\\\hline
6&0.25&2.5&\{0.545,0.449\}, 0.496&\{2.39,2.38\}, \{1.75,2.00\}&\{9.08,8.89\},  \{6.57,8.22\}&\{2.49,2.48\},  \{2.28,2.39\}\\\hline
6&0.5&1&\{0.517,0.468\}, 0.519&\{1.75,1.80\}, \{1.85,1.91\}&\{5.91,6.09\},  \{7.31,7.50\}&\{2.27,2.28\},  \{2.33,2.35\}\\\hline
6&0.5&2&\{0.556,0.497\}, 0.554&\{1.68,1.73\}, \{1.71,1.78\}&\{5.32,5.43\},  \{6.30,6.61\}&\{2.22,2.23\},  \{2.26,2.28\}\\\hline
6&0.5&2.5&\{0.591,0.518\}, 0.579&\{1.67,1.71\}, \{1.60,1.70\}&\{5.24,5.32\},  \{5.56,6.04\}&\{2.21,2.22\},  \{2.19,2.24\}\\\hline
7&0.1&0.2&\{0.134,0.131\}, 0.135&\{6.02,6.03\}, \{7.43,7.38\}&\{64.82,64.99\},  \{110.86,109.31\}&\{3.48,3.48\},  \{3.70,3.69\}\\\hline
7&0.1&0.4&\{0.168,0.163\}, 0.170&\{4.35,4.36\}, \{5.94,5.88\}&\{29.97,30.06\},  \{71.06,69.37\}&\{3.11,3.11\},  \{3.48,3.47\}\\\hline
7&0.1&0.5&\{0.179,0.174\}, 0.182&\{4.06,4.06\}, \{5.55,5.50\}&\{25.37,25.42\},  \{62.14,60.66\}&\{3.03,3.03\},  \{3.41,3.40\}\\\hline
7&0.1&1&\{0.232,0.219\}, 0.231&\{3.54,3.53\}, \{4.28,4.32\}&\{18.36,18.27\},  \{37.12,37.57\}&\{2.84,2.84\},  \{3.15,3.16\}\\\hline
7&0.167&0.333&\{0.200,0.192\}, 0.201&\{4.13,4.15\}, \{4.97,4.96\}&\{29.70,29.86\},  \{49.96,49.47\}&\{3.10,3.10\},  \{3.30,3.30\}\\\hline
7&0.167&0.667&\{0.232,0.223\}, 0.235&\{3.39,3.40\}, \{4.26,4.25\}&\{18.25,18.30\},  \{36.87,36.37\}&\{2.86,2.86\},  \{3.15,3.14\}\\\hline
7&0.167&0.833&\{0.246,0.235\}, 0.248&\{3.27,3.27\}, \{4.02,4.03\}&\{16.64,16.65\},  \{32.89,32.67\}&\{2.81,2.81\},  \{3.09,3.09\}\\\hline
7&0.167&1.67&\{0.334,0.299\}, 0.321&\{3.05,3.03\}, \{2.94,3.10\}&\{14.12,13.94\},  \{17.79,19.52\}&\{2.71,2.70\},  \{2.78,2.83\}\\\hline
7&0.2&0.4&\{0.231,0.221\}, 0.233&\{3.63,3.65\}, \{4.29,4.29\}&\{23.00,23.15\},  \{37.34,37.05\}&\{2.97,2.97\},  \{3.15,3.15\}\\\hline
7&0.2&0.8&\{0.263,0.250\}, 0.265&\{3.11,3.12\}, \{3.76,3.76\}&\{15.52,15.57\},  \{28.75,28.56\}&\{2.78,2.78\},  \{3.02,3.02\}\\\hline
7&0.2&1&\{0.278,0.262\}, 0.279&\{3.02,3.02\}, \{3.55,3.58\}&\{14.45,14.45\},  \{25.77,25.82\}&\{2.74,2.74\},  \{2.97,2.97\}\\\hline
7&0.2&2&\{0.388,0.339\}, 0.366&\{2.86,2.85\}, \{2.51,2.72\}&\{12.77,12.58\},  \{13.09,15.03\}&\{2.66,2.65\},  \{2.63,2.70\}\\\hline
\end{tabular}
\label{tab:met1}
\end{table*}

\end{document}